\documentclass[journal]{IEEEtran}
\ifCLASSINFOpdf
  \usepackage[pdftex]{graphicx}
  \graphicspath{{./figs/}}
\else
 \usepackage[dvips]{graphicx}
 \graphicspath{{./figs/}}
\fi

\usepackage{amsmath,graphicx,amssymb,amsthm,amsfonts,mathrsfs,cite}
\usepackage{mathtools,xcolor}
\usepackage{soul}
\usepackage{color}
\usepackage{algpseudocode,algorithm}
\usepackage{epstopdf}

\newtheorem{theorem}{Theorem}

\newtheorem{assumption}{Assumption}
\long\def\symbolfootnote[#1]#2{\begingroup%
	\def\thefootnote{\fnsymbol{footnote}}\footnote[#1]{#2}\endgroup} 

\begin{document}
\title{Optimal Feature Manipulation Attacks Against Linear Regression}

\author{\IEEEauthorblockN{Fuwei Li, Lifeng Lai, and Shuguang Cui}}

\maketitle
 \symbolfootnote[0]{
 	F. Li, L. Lai are with the Department of Electrical and Computer Engineering, University of California, Davis, CA, 95616. Email:\{fli, lflai\}@ucdavis.edu.
 	S. Cui is currently with the Shenzhen Research Institute of Big Data and Future Network of Intelligence Institute (FNii), the Chinese University of Hong Kong, Shenzhen, China, 518172, and was with the Department of Electrical and Computer Engineering, University of California at Davis, Davis, CA, USA, 95616 (e-mail: robert.cui@gmail.com).
 	The work was partially supported by the National Science Foundation with Grants CNS-1824553, CCF-1717943, and CCF-1908258.
 }

\pagestyle{plain}

\begin{abstract}
In this paper, we investigate how to manipulate the coefficients obtained via linear regression by adding carefully designed poisoning data points to the dataset or modify the original data points. Given the energy budget, we first provide the closed-form solution of the optimal poisoning data point when our target is modifying one designated regression coefficient. We then extend the analysis to the more challenging scenario where the attacker aims to change one particular regression coefficient while making others to be changed as small as possible. For this scenario, we introduce a semidefinite relaxation method to design the best attack scheme. Finally, we study a more powerful adversary who can perform a rank-one modification on the feature matrix. We propose an alternating optimization method to find the optimal rank-one modification matrix. Numerical examples are provided to illustrate the analytical results obtained in this paper.
\end{abstract}

\begin{IEEEkeywords}
Linear regression, adversarial robustness, poisoning attack, non-convex optimization.
\end{IEEEkeywords}

\section{Introduction}
Linear regression plays a fundamental role in machine learning and is used in a wide spectrum of applications \cite{yan2009linear,g2015robustlinear,beamforming2016,chien2009recursive,gusta2002stat}. In linear regression, one assume that there is a simple linear relationship between the explanatory variables and the response variable. The goal of linear regression is to find out the regression coefficients through the methods of ordinary least square (OLS), ridge regression, Lasso \cite{tibshirani1996regression}, etc. Having the regression coefficients learned from the data points, one can predict the response values given the values of the explanatory variables. The regression coefficients also help us explain the variation in the response variable that can be attributed to the variation in the explanatory variables.  It can quantify the strength of the relationship between certain explanatory variables and the response variable. Large magnitude of the regression coefficient usually indicates a strong relationship while small valued regression coefficient means a weak relationship. This is especially true when linear regression is accomplished by the parameter regularized method such as ridge regression and Lasso. In addition, the sign of the regression coefficients indicates whether the value of the response variable  increases or decreases when the value of a explanatory variable changes, which is very important in biologic science~\cite{mcdonald2009handbook}, financial analysis~\cite{barndorff2004econometric}, and environmental science~\cite{ter1993weighted}. 

Machine learning is being used in various applications, including security and safety critical applications such as medical image analysis~\cite{finlayson2019adversarial} and autonomous driving~\cite{sallab2017deep}. 
For these applications, it is important to understand the robustness of machine learning algorithms in adversarial environments~\cite{goodfellow6572explaining,kurakin2016adversarial,goodfellow2018making}. In such an environment, there may exist a malicious adversary who has the full knowledge of the machine learning system and has the ability to observe the whole data points. After seeing the data points, the adversary can add some carefully designed poisoning data points or directly modify the data points so as to corrupt the learning system or leave a backdoor in this system \cite{chen2017targeted}. 

The goal of this paper is to investigate the optimal way to attack linear regression methods. In the considered linear regression system, there exists an adversary who can observe the whole dataset and then inject carefully designed poisoning data points or directly modify the original dataset in order to manipulate the regression coefficients. The manipulated regression coefficients can later be used by the adversary as a backdoor of this learning system or mislead our interpretation of the linear regression model. For example, by changing a large magnitude regression coefficient to be small, it makes us believe that its corresponding explanatory variable is irrelevant. Similarly, the adversary can change the magnitude of a regression coefficient to a larger value to increase its importance. Furthermore, changing the sign of a regression coefficient can also lead us to misinterpret the correlation between its explanatory variable and the response variable.

Depending on the objective of the adversary and the way the adversary changes the regression coefficients, we have different problem formulations. We first consider a scenario where the adversary tries to manipulate one specific regression coefficient by adding one carefully designed poisoning data point that has limited energy budget to the dataset. 
We show that finding the optimal attack data point is equivalent to solve an optimization problem where the objective function is a ratio of two quadratic functions and the constraint is a quadratic inequality. Even though this type of problem is non-convex in general, our particular problem has a hidden convex structure. With the help of this convex structure, we further convert the optimization problem into a quadratic constrained quadratic program (QCQP). Since it is known that strong duality exists in this problem \cite{boyd2004convex}, we are able to identify its closed-from optimal solutions from its Karush-Kuhn-Tucker (KKT) conditions. 

We next consider a more sophisticated objective where the attacker aims to change one particular regression coefficient while making others to be changed as small as possible. 
We show that the problem of finding the optimal attack data point is equivalent to solving an optimization problem where the objective function is a ratio of two fourth order multivariate polynomials and the constraint is quadratic. This optimization problem is much more complex than the optimization above. We introduce a semidefinite relaxation method to solve this problem. The numerical examples show that we can find the global optimal solutions with very low relaxation order. Hence, the complexity of this method is low in practical problems.

Finally, we consider a more powerful adversary who can directly modify the existing data points in the feature matrix. Particularly, we consider a rank-one modification attack \cite{li2020trans}, where the attacker carefully designs a rank-one matrix and adds it to the existing data matrix. A rank-one modification attack is general enough to capture most of the common modifications, such as modifying one feature, deleting or adding one data point, changing one entry of the data matrix, etc. Hence, studying the rank-one modification provides us universal bounds on these kinds of attacks. By leveraging the rank-one structure, we develop an alternating optimization method to find the optimal modification matrix. We also prove that the solution obtained by the proposed optimization method is one of the critical points of the optimization problem. 


Our study is related to several recent works on the adversarial machine learning. For example, Pimentel-Alarc\'{o}n et al. studied how to add one adversarial data point in order to maximize the error of the subspace estimated by principal component \cite{pimentel2017adversarial},
Li et al. studied the adversarial robustness of subspace learning problem \cite{li2020trans}, and Alfeld et al. considered how to use poisoning data points to attack the auto regression model \cite{adregression}. The work that is most relevant to our paper is \cite{mei2015using}, where the authors develop a bi-level optimization framework to design the attack matrix. \cite{mei2015using} further proposes to use the projected gradient descent method to solve the bi-level optimization problem. However, a general bi-level problem is known to be NP hard and solving it depends on the convexity of the lower level problem. In addition, the convergence of projected gradient descent for non-convex problem is not clear. Compared with \cite{mei2015using}, we obtain the global optimal solution for the case with adding one poisoning data point, and we also prove that the proposed alternating optimization method converges to one of the critical points for the case where the attacker can perform rank-one modification attack. Furthermore, for the projected gradient descent method, different datasets may need different parameters, which makes the parameters of this algorithm hard to tune. By contrast, we provide closed-form solution for the case with adding one poisoning data point to attack one of the regression coefficients, and the designed alternating optimization method for the case of rank-one attack does not need parameter tuning. Furthermore, compared with the projected gradient descent method, our alternating optimization method provides smaller objective values, faster convergence rate, and more stable behavior. 

The remainder of this paper is organized as follows. In Section~\ref{sec:one-data-attack}, we consider the scenario where the attack add one carefully designed poisoning data point to the dataset. 
In Section~\ref{sec:atk-r1}, we investigate the rank-one attack strategy. Numerical examples are provided in Section~\ref{sec:ne} to illustrate the results we obtained in this paper. 
Finally, we provide concluding remarks in Section~\ref{sec:conclusion}. 

\section{Attacking with one adversarial data point}\label{sec:one-data-attack}

In this section, we consider the scenario where the attacker can add one carefully crafted data point to the existing dataset. We will extend the analysis to the case with more sophisticated attacks in Section~\ref{sec:atk-r1}.

\subsection{Problem formulation}
Consider a dataset with $n$ data samples, $\{y_i, \mathbf{x}_i\}_{i=1}^n$, where $y_i$ is the response variable, $\mathbf{x}_i \in \mathbb{R}^m$ is the feature vector, where each component of $\mathbf{x}_i$ represents an explanatory variable. 
In this section, we consider an adversarial setup in which the adversary first observes the the whole dataset $\{\mathbf{y},\mathbf{X}\}$, in which $\mathbf{y} := [y_1,\,y_2,\,\ldots,y_n]^\top$ and $\mathbf{X} := [\mathbf{x}_1,\,\mathbf{x}_2,\ldots,\mathbf{x}_n]^\top$, and then carefully designs an adversarial data point, $\{y_0, \mathbf{x}_0\}$, and adds it into the existing data samples. After inserting this adversarial data point, we have the poisoned dataset $\{\hat{\mathbf{y}},\hat{\mathbf{X}\}}$, where $\hat{\mathbf{y}}:=[y_0,\,y_1,\,y_2,\ldots,y_n]^\top$,  $\hat{\mathbf{X}}:=[\mathbf{x}_0,\,\mathbf{x}_1,\mathbf{x}_2,\ldots,\mathbf{x}_n]^\top$.  

From the dataset, we intend to learn a linear regression model. From the poisoned dataset, the learned model is obtained by solving
\begin{align}\label{opt:mod-reg}
 \underset{\boldsymbol{\beta}}{\text{argmin}}:  \|\hat{\mathbf{y}} - \hat{\mathbf{X}}\boldsymbol{\beta}\|^2,
\end{align}
where $\|\cdot\|$ denotes the $\ell_2$ norm for a vector and the induced $2$-norm for a matrix throughout this paper.
Let $\hat{\boldsymbol{\beta}}$ be the optimal solution to problem~\eqref{opt:mod-reg}. The goal of the adversary is to minimize some objective function, $f(\hat{\boldsymbol{\beta}})$, by carefully designing the adversarial data point. The form of $f(\hat{\boldsymbol{\beta}})$ depends on the specific goal of the attacker. For example, the attacker can try to reduce the importance of feature $i$ by setting $f(\hat{\boldsymbol{\beta}})= |\hat{\beta}_i|$, in which $\hat{\beta}_i$ is the $i$th component of $\hat{\boldsymbol{\beta}}$. Or the attacker can try to increase the importance of feature $i$ by setting $f(\hat{\boldsymbol{\beta}})= -|\hat{\beta}_i|$.  
To make the problem meaningful, in this paper, we impose the energy constraint on the adversarial data point and use the $\ell_2$ norm to measure its energy. With the objective $f(\hat{\boldsymbol{\beta}})$ and the energy constraint of the adversary data point, our problem can be formulated as 
\begin{align}\label{pro:form}
    \min_{\|[\mathbf{x}_0^\top, y_0]\|\le \eta}: &\quad f(\hat{\boldsymbol{\beta}}) \\\nonumber
    \text{s.t.}& \quad \hat{\boldsymbol{\beta}} = \underset{\boldsymbol{\beta}}{\text{argmin}}: \|\hat{\mathbf{y}} - \hat{\mathbf{X}}\boldsymbol{\beta}\|^2,
\end{align}
where $\eta$ is the energy budget. 
This is a complicated bi-level optimization problem. The objective function, $f(\hat{\boldsymbol{\beta}})$, depends on the poisoning data point, $\{\mathbf{x}_0, y_0\}$, not in a direct way, but through a lower level optimization problem. What makes this problem even harder is the complication of the objective function. 
Depending on the goal of the adversary, the objective can be in various of forms. In the following two subsections, we will discuss two important objectives and their solutions, respectively. The methods and insights obtained from these two cases could be then be extended to cases with other objectives.

\subsection{Attacking one regression coefficient}\label{sec:atk-one}
In this subsection, the goal of the adversary is to design the adversarial data point $\{y_0,\mathbf{x}_0\}$ to decrease (or increase) the importance of a certain explanatory variable. If the goal is to decrease the importance of explanatory variable $i$, we can set $f(\hat{\boldsymbol{\beta}})= |\hat{\beta}_i|$, and the optimization problem can be written as 
\begin{align} \label{opt:orig-min}
\underset{\|[\mathbf{x}_0^\top,\, y_0]\|_2 \le \eta}{\text{min}}:&\quad |\hat{\beta}_i| \\\nonumber
\text{s.t.} & \quad \hat{\boldsymbol{\beta}} = \underset{\boldsymbol{\beta}}{\text{argmin}}: \| \hat{\mathbf{y}} - \hat{\mathbf{X}}\boldsymbol{\beta}\|^2.
\end{align}


Similarly, if the goal of the adversary is to increase the importance of explanatory variable $i$, we can set $f(\hat{\boldsymbol{\beta}})= - |\hat{\beta}_i|$, and we have the optimization problem
\begin{align} \label{opt:orig-max}
\underset{\|[\mathbf{x}_0^\top,\, y_0]\| \le \eta}{\text{min}}:& \quad -|\hat{\beta}_i|\\
\text{s.t.} & \quad \hat{\boldsymbol{\beta}} = \underset{\boldsymbol{\beta}}{\text{argmin}}: \| \hat{\mathbf{y}} - \hat{\mathbf{X}}\boldsymbol{\beta}\|^2.\nonumber
\end{align}



To solve the complicated bi-level optimization problems problems~\eqref{opt:orig-min} and~\eqref{opt:orig-max}, we first solve the following two optimization problems
\begin{align}\label{opt:min-beta}
\underset{\|[\mathbf{x}_0^\top,\, y_0]\| \le \eta}{\min}&: \hat{\beta}_i \\
\text{s.t.}& \quad \hat{\boldsymbol{\beta}} = \underset{\boldsymbol{\beta}}{\min}: \|\hat{\mathbf{y}}- \hat{\mathbf{X}}\boldsymbol{\beta} \|^2,
\label{opt:bilevel-2}
\end{align}

and 
\begin{align}\label{opt:obj-max}
\underset{\|[\mathbf{x}_0^\top,\, y_0]\| \le \eta}{\max}&: \hat{\beta}_i \\
\text{s.t.}& \quad \hat{\boldsymbol{\beta}} = \underset{\boldsymbol{\beta}}{\min}: \|\hat{\mathbf{y}}- \hat{\mathbf{X}}\boldsymbol{\beta} \|^2.
\end{align}
It is easy to check that the solutions to problems~\eqref{opt:orig-min} and \eqref{opt:orig-max} can be obtained from the solutions to problem~\eqref{opt:min-beta} and~\eqref{opt:obj-max}. In particular, let $(\hat{\beta}_i^*)_{\min}$ and $(\hat{\beta}_i^*)_{\max}$ be optimal values of problem~\eqref{opt:min-beta} and \eqref{opt:obj-max} respectively. Then, if $\hat{\beta}_i \ge 0$, we can check that $\max\{0, (\hat{\beta}_i^*)_{\min}\}$ and $\max\{|(\hat{\beta}_i^*)_{\min}|,\, |(\hat{\beta}^*_i)_{\max}|\}$ are the solutions to problem~\eqref{opt:orig-min} and \eqref{opt:orig-max} respectively. Similar arguments can be made if $\hat{\beta}_i< 0$.  

In the following, we will focus on solving the minimization problem~\eqref{opt:min-beta}. The solution to the maximization problem~\eqref{opt:obj-max} can be obtained by using a similar approach. To solve this bi-level optimization problem, we can first solve the optimization problem in the subjective. Problem~\eqref{opt:bilevel-2} is just an ordinary least squares problem, which has a simple closed-form solution: $\hat{\boldsymbol{\beta}} = (\hat{\mathbf{X}}^\top\hat{\mathbf{X}})^{-1}\hat{\mathbf{X}}^\top \hat{\mathbf{y}}$. Substitute in $\hat{\mathbf{X}} = [\mathbf{x}_0,\mathbf{X}^\top]^\top$ and $\hat{\mathbf{y}} =[y_0,\mathbf{y}^\top]^\top$, and we have 
\begin{align*}
    \hat{\boldsymbol{\beta}} = (\mathbf{X}^\top\mathbf{X}+\mathbf{x}_0 \mathbf{x}_0^\top)^{-1}[\mathbf{x}_0,\mathbf{X}^\top][y_0,\mathbf{y}^\top]^\top.
\end{align*}
According to the Sherman-Morrison formula~\cite{horn2012matrix}, we have 
\begin{eqnarray}
  && \hspace{-15mm} (\mathbf{X}^\top\mathbf{X}+ \mathbf{x}_0\mathbf{x}_0^\top)^{-1} \nonumber\\
    &&= (\mathbf{X}^\top\mathbf{X})^{-1}  -  \frac{(\mathbf{X}^\top \mathbf{X})^{-1} \mathbf{x}_0\mathbf{x}_0^\top(\mathbf{X}^\top \mathbf{X})^{-1}}{1+\mathbf{x}_0^\top (\mathbf{X}^\top\mathbf{X})^{-1} \mathbf{x}_0}.\nonumber
\end{eqnarray}
The inverse of $\mathbf{X}^\top\mathbf{X}+\mathbf{x}_0\mathbf{x}_0^\top$ always exists because $1+\mathbf{x}_0^\top (\mathbf{X}^\top\mathbf{X})^{-1} \mathbf{x}_0 \neq 0$. Plug this inverse in the expression of $\hat{\boldsymbol{\beta}}$, we get 
\begin{align}
    \hat{\boldsymbol{\beta}} 
    = \boldsymbol{\beta}_0 + \frac{\mathbf{A}\mathbf{x}_0(y_0-\mathbf{x}_0^\top\boldsymbol{\beta}_0)}{1+\mathbf{x}_0^\top\mathbf{A}\mathbf{x}_0}, \label{eq:beta_hat}
\end{align}
where 
\begin{eqnarray}
\mathbf{A} &=& (\mathbf{X}^\top\mathbf{X})^{-1}, \label{eq:A}\\
\boldsymbol{\beta}_0&=& ({\mathbf{X}}^\top{\mathbf{X}})^{-1}{\mathbf{X}}^\top {\mathbf{y}}.\label{eq:beta0}
\end{eqnarray} 
We can observe that $\boldsymbol{\beta}_0$ is the coefficient that is obtained from the clean data. Problem~\eqref{opt:min-beta} is equivalent to 
\begin{align}\label{opt:r-q-orig}
    \underset{\mathbf{x}_0,y_0}{\min}&: \frac{\mathbf{a}^\top \mathbf{x}_0(y_0 - \mathbf{x}_0^\top\boldsymbol{\beta}_0)}{1+\mathbf{x}_0^\top\mathbf{A}\mathbf{x}_0} \\
    \text{s.t.}& \quad \|[\mathbf{x}_0^\top,\, y_0] \| \le \eta,
\end{align}
where $\mathbf{a}$ is the $i$th column of $\mathbf{A}$. The optimization problem~\eqref{opt:r-q-orig} is the ratio of two quadratic functions with a quadratic constraint. To further simplify this optimization problem, we can write our objective and subjective in a more compact form by performing variable change: $\mathbf{u} = [\mathbf{x}_0^\top,\, y_0]^\top$. Using this compact representation, the optimization problem~\eqref{opt:r-q-orig} can be written as
\begin{align} \label{opt:qcrq-orig}
    \underset{\mathbf{u}}{\min:}&\quad 
    \frac{\frac{1}{2}\mathbf{u}^\top \mathbf{H} \mathbf{u}}
    {1+\mathbf{u}^\top
    \begin{bsmallmatrix}
     &\mathbf{A} &\mathbf{0}\\
     &\mathbf{0} & 0
    \end{bsmallmatrix}
    \mathbf{u}} \\\nonumber
    \text{s.t.}&\quad \mathbf{u}^\top \mathbf{u} \le \eta^2,
\end{align}
in which
\begin{align}\label{eq:H}
    \mathbf{H} = 
    \begin{bmatrix}
    -\mathbf{a}\boldsymbol{\beta}_0^\top -\boldsymbol{\beta}_0\mathbf{a}^\top & \mathbf{a}\\
    \mathbf{a}^\top & 0 
    \end{bmatrix}. 
\end{align}

~\eqref{opt:qcrq-orig} is a non-convex optimization problem. To solve this problem, we employ the technique introduced in~\cite{beck2010minimizing}. We first perform variable change $\mathbf{u} = \frac{\mathbf{\mathbf z}}{s}$ by introducing variable $\mathbf{z}$ and scalar $s$. Inserting this into problem~\eqref{opt:qcrq-orig}, adding constraint $1$ to the denominator of the objective and moving it to the subjective, we have a new optimization problem
\begin{align} \label{opt:qcrq}
    \underset{\mathbf{z},s}{\min}:\quad& 
    \frac{1}{2}\mathbf{z}^\top \mathbf{H} \mathbf{z}\\
    \text{s.t.}\quad &
    s^2+\mathbf{z}^\top \begin{bsmallmatrix}
     &\mathbf{A}  & \mathbf{0}\\
     &\mathbf{0}  & 0
    \end{bsmallmatrix}
    \mathbf{z}=1, \label{eq:s}\\
    &\mathbf{z}^\top \mathbf{z} \le s^2\eta^2\label{ineq:z}.
\end{align}
To validate the equivalence between problem~\eqref{opt:qcrq-orig} and \eqref{opt:qcrq}, we only need to check if the optimal value of problem~\eqref{opt:qcrq-orig} is less than the optimal value of problem~\eqref{opt:qcrq} when $s=0$~\cite{beck2010minimizing}. Firstly, since $\mathbf{H}$ is not positive semi-definite (which will be shown later), the optimal value of problem~\eqref{opt:qcrq-orig} is less than zero. Secondly, when $s=0$, the optimal value of problem~\eqref{opt:qcrq} is zero, which is apparently larger than the optimal value of problem~\eqref{opt:qcrq-orig}. Therefore, the two problems are equivalent. 

To solve problem~\eqref{opt:qcrq}, we substitute $s^2$ in equation~\eqref{eq:s} for that in equation~\eqref{ineq:z} and obtain
\begin{align}\label{opt:qcrq-simp}
    \underset{\mathbf{z}}{\min}:\quad& 
    \frac{1}{2}
    \mathbf{z}^\top \mathbf{H} \mathbf{z}\\
    \text{s.t.}\quad    &\mathbf{z}^\top\left( \mathbf{I}+\eta^2
    \begin{bsmallmatrix}
     &\mathbf{A} &\mathbf{0}\\
     &\mathbf{0} & 0
    \end{bsmallmatrix}
    \right)\mathbf{z}\le\eta^2.\label{ineq:z2}
\end{align}
Note that $\mathbf{H}$ is not positive semi-definite; hence problem~\eqref{opt:qcrq-simp} is not a standard convex QCQP problem~\cite{boyd2004convex}. However, it is proved that strong duality holds for this type of problem~\cite{boyd2004convex,konar2017fast}. Hence, to solve this problem, we can start by investigating its KKT necessary conditions. The Lagrangian of problem~\eqref{opt:qcrq-simp} is
\begin{align*}
    \mathcal{L}(\mathbf{z}, \lambda) =
    \frac{1}{2}
    \mathbf{z}^\top \mathbf{H} \mathbf{z} 
    +\lambda \left(
    \mathbf{z}^\top\left( \mathbf{I}+\eta^2
    \begin{bsmallmatrix}
     &\mathbf{A} &\mathbf{0}\\
     &\mathbf{0} & 0
    \end{bsmallmatrix}
    \right)\mathbf{z} - \eta^2
    \right),
\end{align*}
where $\lambda$ is the dual variable.
According to the KKT conditions, we have 
\begin{align}
     \left(
    \mathbf{H}  + \lambda \mathbf{D} \right)\mathbf{z} &=\mathbf{0}, \label{kkt:stationary}\\
    \frac{1}{2}\mathbf{z}^\top \mathbf{D} \mathbf{z}  &\le \eta^2, \label{kkt:prim}\\
    \lambda \left(
    \frac{1}{2}\mathbf{z}^\top \mathbf{D} \mathbf{z} - \eta^2 \right) & =  0, \label{kkt:slack}\\
    \lambda &\ge 0, \label{kkt:dual}
\end{align}
where 
\begin{align}\label{eq:D}
    \mathbf{D} = 2
    \left( 
    \mathbf{I}+\eta^2
    \begin{bmatrix}
     \mathbf{A} &\mathbf{0}\\
     \mathbf{0} & 0
    \end{bmatrix}
    \right). 
\end{align}

By inspecting the complementary slackness condition~\eqref{kkt:slack}, we consider two cases based on the value of $\lambda$. \\
\noindent
\textbf{Case 1}: $\lambda =0$. In this case, we must have $\mathbf{H}\mathbf{z}=\mathbf{0}$. As a result, the objective value of~\eqref{opt:qcrq-simp} is zero, which contradicts the fact that the optimal value should be negative. Hence, this case is not possible.  \\
\noindent
\textbf{Case 2}: $\lambda > 0 $. In this case, equality in~\eqref{kkt:prim} must hold. According to the stationary condition~\eqref{kkt:stationary}, if the matrix $\mathbf{H}+\lambda\mathbf{D}$ is full rank, we must have $\mathbf{z}=\mathbf{0}$, for which equality in~\eqref{kkt:prim} cannot hold. Hence, $\mathbf{H}+\lambda \mathbf{D}$ is not full-rank and we have
$\det(\mathbf{H}+\lambda\mathbf{D}) = 0$.
As $\mathbf{D}$ is positive definite, we also have
\begin{align}
    \det(\mathbf{D}^{-1/2}\mathbf{H}\mathbf{D}^{-1/2}+\lambda \mathbf{I}) = 0.
\end{align}
Since $\lambda > 0$, this equality tells us that $-\lambda$ belongs to one of the negative eigenvalues of $\mathbf{D}^{-1/2}\mathbf{H}\mathbf{D}^{-1/2}$. In the following, we will show that $\mathbf{D}^{-1/2}\mathbf{H}\mathbf{D}^{-1/2}$ has one and only one negative eigenvalue. 

By definition, $\mathbf{D}$ is a block diagonal matrix. Hence, its inverse is also block diagonal. Let us define
$\mathbf{D}^{-1/2} = \text{diag}\{\mathbf{G}, g\}$,
where  $\mathbf{G} = 1/\sqrt{2}(\mathbf{I}+\eta^2\mathbf{A})^{-1/2}$ and $g = 1/\sqrt{2}$. Thus, we have 
\begin{align*}
    \mathbf{D}^{-1/2}\mathbf{H}\mathbf{D}^{-1/2} 
    = 
    \begin{bmatrix}
    &-\mathbf{c}\mathbf{h}^\top - \mathbf{h}\mathbf{c}^\top & g\mathbf{c} \\
    &g\mathbf{c}^\top & 0 
    \end{bmatrix},
\end{align*}
where $\mathbf{c} = \mathbf{G}\mathbf{a}$ and $\mathbf{h}= \mathbf{G}\boldsymbol{\beta}_0$. Define $\xi$ as the eigenvalue of $\mathbf{D}^{-1/2}\mathbf{H}\mathbf{D}^{-1/2}$, and compute its eigenvalues by computing the characteristic polynomial: 
\begin{align*}
     &\det\left(\xi\mathbf{I}-\mathbf{D}^{-1/2}\mathbf{H}\mathbf{D}^{-1/2}\right)\\
     &= \xi^{m-1}\left(
    \xi^2+2\xi \mathbf{c}^\top\mathbf{h} + \mathbf{c}^\top\mathbf{h}\mathbf{h}^\top\mathbf{c} - g^2\mathbf{c}^\top\mathbf{c} - \mathbf{c}^\top\mathbf{c}\mathbf{h}^\top\mathbf{h}
    \right).
\end{align*}
Thus, the eigenvalues of $\mathbf{D}^{-1/2}\mathbf{H}\mathbf{D}^{-1/2}$ are $\xi =0$ ((m-1) multiplications) and $\xi = -\mathbf{c}^\top\mathbf{h} \pm \|\mathbf{c}\| \sqrt{g^2+\mathbf{h}^\top\mathbf{h}}$. Since $\|\mathbf{c}\|\sqrt{g^2 + \mathbf{h}^\top \mathbf{h}} > |\mathbf{c}^\top \mathbf{h}|$, the eigenvalues of $\mathbf{D}^{-1/2}\mathbf{H}\mathbf{D}^{-1/2}$ satisfy: 
$\xi_{m+1}<0,\quad \xi_m=\xi_{m-1}=\cdots=\xi_2=0,\quad  \xi_1 >0$.
Now, it is clear that $\mathbf{D}^{-1/2}\mathbf{H}\mathbf{D}^{-1/2}$ has one and only one negative eigenvalue and one positive eigenvalue, respectively. Thus, we have $\lambda = - \xi_{m+1}$. Assume $\boldsymbol{\nu}_1$ and $\boldsymbol{\nu}_{m+1}$  are two eigenvectors corresponding to eigenvalues $\xi_1$ and $\xi_{m+1}$. Through simple calculation, we have 
\begin{align} \label{v:eig}
    \boldsymbol{\nu}_i = k_i
    \left[
    -\frac{\mathbf{c}^\top \mathbf{h}+\xi_i}{\mathbf{c}^\top\mathbf{c}}\mathbf{c}^\top+\mathbf{h}^\top, 
    \frac{g\mathbf{c}^\top}{\xi_i}\left(-\frac{\mathbf{c}^\top \mathbf{h}+\xi_i}{\mathbf{c}^\top\mathbf{c}}\mathbf{c} +\mathbf{h}\right )
    \right]^\top,
\end{align}
where $i = 1,\,m+1$ and scalar $k_{i}$ is the normalization constant to guarantee the eigenvectors to be of unit length.  According to \eqref{kkt:stationary}, we have
\begin{align*}
    \left(\mathbf{H}+\lambda\mathbf{D}\right)\mathbf{z}=
    \mathbf{D}^{1/2}\left(\mathbf{D}^{-1/2}\mathbf{H}\mathbf{D}^{-1/2}+\lambda\mathbf{I}\right)\mathbf{D}^{1/2}\mathbf{z}=0;
\end{align*}
thus the solution to problem~\eqref{opt:qcrq-simp} is 
\begin{align}
    \mathbf{z}^* = k \cdot \mathbf{D}^{-1/2}\boldsymbol{\nu}_{m+1}.
    \label{sol:z}
\end{align}
Since $\frac{1}{2} \mathbf{z}^\top \mathbf{D}\mathbf{z} = \eta^2$, we have $k=\sqrt{2}\eta$.
Having the expression of the optimal $\mathbf{z}^*$, we can then compute $s$ according to equation~\eqref{eq:s}: 
\begin{align}\label{sol:s}
    s = \pm
    \sqrt{
    1 - (\mathbf{z}_{1:m}^*)^\top
    \mathbf{A}\,
    \mathbf{z}_{1:m}^*
    },
\end{align}
where $\mathbf{z}_{1:m}^*$ is the vector that comprises the first $m$ elements of $\mathbf{z}^*$. 
Hence, the corresponding solution to problem~\eqref{opt:r-q-orig} is 
\begin{align} \label{sol:xy}
    \mathbf{x}_0^* = \mathbf{z}_{1:m}^*/s,\quad 
    y_0^* = {z}_{m+1}^*/s.
\end{align}

We now compute the optimal value of problem~\eqref{opt:qcrq}. Since our objective function is $\frac{1}{2}(\mathbf{z}^*)^\top\mathbf{H}\mathbf{z}^*$, substituting $\mathbf{z}^*$ in \eqref{sol:z} leads to the objective value:
$\eta^2\boldsymbol{\nu}_{m+1}^\top \mathbf{D}^{-1/2}\mathbf{H}\mathbf{D}^{-1/2}\boldsymbol{\nu}_{m+1}.$
Since $\boldsymbol{\nu}_{m+1}^\top \mathbf{D}^{-1/2}\mathbf{H}\mathbf{D}^{-1/2}\boldsymbol{\nu}_{m+1} = \xi_{m+1}$, our optimal objective value is $\eta^2\xi_{m+1}$.

\begin{algorithm}[t!]
	\caption{Optimal Adversarial Data Point Design}\label{alg:ad-attack}
	\begin{algorithmic}[1]
		\State \textbf{Input}: the data set, $\{y_i, \mathbf{x}_i\}_{i=1}^n$, energy budget $\eta$, and the index of feature to be attacked. 
		\State \textbf{Steps}:
		\State compute $\mathbf{A}$ according to equation~\eqref{eq:A}, compute $\boldsymbol{\beta}_0$ according to~\eqref{eq:beta0}. 
		\State compute $\mathbf{H}$ and $\mathbf{D}$ according to~\eqref{eq:H} and~\eqref{eq:D}, respectively. 
		
		\State  compute the last eigenvalue, $\xi_{m+1}$, of $\mathbf{D}^{-1/2}\mathbf{H}\mathbf{D}^{-1/2}$ and its corresponding eigenvector according to~\eqref{v:eig}.
		\State design the adversarial data point, $\{\mathbf{x}_0, y_0\}$, according to equations~\eqref{sol:z}, \eqref{sol:s}, and \eqref{sol:xy}. 
		\State \textbf{Output}: return the optimal adversarial data point $\{\mathbf{x}_0, y_0\}$ and the optimal value $\eta^2\xi_{m+1}+(\boldsymbol{\beta}_0)_i$.
	\end{algorithmic}
\end{algorithm}

Following similar analysis as above, we can find the optimal $\mathbf{z}^*$ for problem~\eqref{opt:obj-max}, which is $\mathbf{z}^* = \sqrt{2}\eta\mathbf{D}^{-1/2}\boldsymbol{\nu}_1$. Also, we can compute the optimal $\mathbf{x}_0^*$ and $y_0^*$ according to equation~\eqref{sol:xy} and its optimal objective value, which is $\eta^2\xi_1$. 

In summary, the optimal values for problems~\eqref{opt:min-beta} and \eqref{opt:obj-max} are $\eta^2\xi_{m+1}+(\boldsymbol{\beta}_0)_i$ and $\eta^2\xi_1+(\boldsymbol{\beta}_0)_i$ respectively. We have summarized the process to design the optimal adversarial data point in Algorithm~\ref{alg:ad-attack} with respect to objective~\eqref{opt:min-beta} and the process with respect to objective~\eqref{opt:obj-max} can be obtained accordingly. 
Based on our optimal values of problems~\eqref{opt:min-beta} and \eqref{opt:obj-max}, we can further decide the optimal values of problems~\eqref{opt:orig-min} and \eqref{opt:orig-max} as discussed at the beginning of this section.

Moreover, if we use the ridge regression method in linear regression, there is only a slight difference in the matrix $\mathbf{A}$ in problem~\eqref{opt:r-q-orig} and the whole analysis remains the same.

\subsection{Attacking with small changes of other regression coefficients}\label{sec:atk-multi}
In Section~\ref{sec:atk-one}, we have discussed how to design the adversarial data points to attack one specific regression coefficient in order to enhance or reduce the importance of its corresponding explanatory variable. However, as we only focus on one particular regression coefficient, other regression coefficients may also be changed by the attack sample. In this subsection, we consider a more complex objective function, where we aim to make the changes to other regression coefficients to be as small as possible while attacking one of the regression coefficients.  

Suppose our objective is to minimize the $i$th regression coefficient (the scenario of maximize the $i$th regression coefficient can be solved using similar approach), i.e., to minimize $\|\hat{\beta}_i\|^2$. At the same time, we would also like to minimize the changes to the rest of the regression coefficients, i.e., to minimize $\|\boldsymbol{\beta}_0^{-i}- \hat{\boldsymbol{\beta}}^{-i}\|^2$, where $\boldsymbol{\beta}_0^{-i}$ contains the original regression coefficients except its $i$th element and $\hat{\boldsymbol{\beta}}^{-i}$ consists of all the elements of the regression coefficients excluding the $i$th one after we insert the adversarial data sample. Combine the two objectives, we have our new objective function $$f(\hat{\boldsymbol{\beta}})=\frac{1}{2}\left\|\boldsymbol{\beta}_0^{-i} - \hat{\boldsymbol{\beta}}^{-i}\right\|^2 + \frac{\lambda}{2}\left\|\hat{\beta}_i\right\|^2 ,$$ where $\lambda$ is the trade-off parameter. The larger the $\lambda$ is, the more effort will be made to keep the $i$th regression coefficient small. Similarly, a negative $\lambda$ means the adversary attempts to make the magnitude of the $i$th regression coefficient large. Again, we assume that the attack energy budget is $\eta$. As the result, we have the following optimization problem
\begin{align}\label{prob:gener-orig}
    \underset{\left\|[\mathbf{x}_0^\top, y_0]\right\|\le \eta }{\text{min}}: 
    & \quad \frac{1}{2}\left\|\boldsymbol{\beta}_0^{-i} - \hat{\boldsymbol{\beta}}^{-i}\right\|^2 + \frac{\lambda}{2}\left\|\hat{\beta}_i\right\|^2 \\ \nonumber
    \text{s.t.}&\quad 
    \hat{\boldsymbol{\beta}} = \underset{\boldsymbol{\beta}}{\text{argmin}}: \|\hat{\mathbf{y}} - \hat{\mathbf{X}}\boldsymbol{\beta} \|^2.
\end{align}
As the objective function is a quadratic function with respect to $\hat{\boldsymbol{\beta}}$, we can write it in a more compact form: $\frac{1}{2}(\hat{\boldsymbol{\beta}}-\mathbf{d})^\top \mathbf{\Lambda} (\hat{\boldsymbol{\beta}}-\mathbf{d})$, where $\mathbf{d}=[\beta_0^1, \beta_0^2, \ldots, \beta_0^{i-1}, 0, \beta_0^{i+1}, \ldots, \beta_0^m]^\top$ and $\mathbf{\Lambda}=\text{diag}(1, 1,\ldots, \lambda, \ldots, 1)$. With this compact form, our optimization problem can be written as \begin{align}\label{prob:quad}
    \underset{\left\|[\mathbf{x}_0^\top, y_0]\right\|\le \eta }{\text{min}}: 
    & \quad \frac{1}{2}(\hat{\boldsymbol{\beta}} - \mathbf{d} )^\top \mathbf{\Lambda}  (\hat{\boldsymbol{\beta}} - \mathbf{d} ) \\ \nonumber
    \text{s.t.}&\quad 
    \hat{\boldsymbol{\beta}} = \underset{\boldsymbol{\beta}}{\text{argmin}}: \|\hat{\mathbf{y}} - \hat{\mathbf{X}}\boldsymbol{\beta} \|^2.
\end{align}
To solve this problem, same as in the previous subsection, we start by solving the lower level optimization problem. Since we have the same lower level problem as in \eqref{opt:min-beta}, substitute $\hat{\boldsymbol{\beta}}$ in the objective with the expression~\eqref{eq:beta_hat}, and we have the one level optimization problem
\begin{align*} 
    \underset{\mathbf{x}_0, y_0}{\min}:& \quad \frac{1}{2} 
    \mathbf{g}^\top 
    \mathbf{\Lambda}
    \mathbf{g} \\\nonumber
    \text{s.t.}& \quad \left\|[\mathbf{x}_0^\top, y_0] \right\| \le \eta,
\end{align*}
where 
$
    \mathbf{g} = \frac{\mathbf{A}\mathbf{x}_0(y_0 -\mathbf{x}_0^\top \boldsymbol{\beta}_0)}{1+\mathbf{x}_0^\top \mathbf{A}\mathbf{x}_0} - \mathbf{b},
$
and $\mathbf{b}=\mathbf{d} - \boldsymbol{\beta}_0$ with $\mathbf{A}$ and $\boldsymbol{\beta}_0$ defined in~\eqref{eq:A} and~\eqref{eq:beta0} respectively. To further simplify our problem, let us define
\begin{align}
    \mathbf{A}_1 = [\mathbf{A}, \mathbf{0}], \,
    \mathbf{A}_2 = 
    \begin{bmatrix}
    \mathbf{A} & \mathbf{0} \\
    \mathbf{0} & 0
    \end{bmatrix},\, 
    \mathbf{c} = 
    \begin{bmatrix}
    -\boldsymbol{\beta}_0 \\
    1
    \end{bmatrix}, \,
    \mathbf{z} =
    \begin{bmatrix}
    \mathbf{x}_0 \\
    y_0
    \end{bmatrix},
    \label{eqs:a1a1}
\end{align}
where $\mathbf{A}_1 \in \mathbb{R}^{m\times (m+1)}$ and $\mathbf{A}_2 \in \mathbb{R}^{(m+1)\times(m+1)}$. With the new defined variables, we can write our problem more compactly as:
\begin{align} 
\label{opt:s23}
    \underset{\mathbf{z}}{\min}:& \quad \frac{1}{2}
    \left( \frac{\mathbf{A}_1\mathbf{z}\mathbf{c}^\top\mathbf{z}}{1+\mathbf{z}^\top \mathbf{A}_2\mathbf{z}} - \mathbf{b} \right)^\top \mathbf{\Lambda} 
     \left( \frac{\mathbf{A}_1\mathbf{z}\mathbf{c}^\top\mathbf{z}}{1+\mathbf{z}^\top \mathbf{A}_2\mathbf{z}} - \mathbf{b} \right) \\ \nonumber
    \text{s.t.}& \quad \|\mathbf{z}\| \le \eta. 
\end{align}
Since the objective is a ratio of two quartic functions, similar to the process we carried out from \eqref{opt:qcrq-orig} to \eqref{opt:qcrq}, we perform variable change $\mathbf{z} = \frac{\mathbf{w}}{s}$ by introducing the new variable $\mathbf{w}$ and scalar $s$. Insert it into problem~\eqref{opt:s23} and follow the same argument we have made to transform problem~\eqref{opt:qcrq-orig} to problem~\eqref{opt:qcrq}, problem~\eqref{opt:s23} is equivalent to the following problem
\begin{align}\label{opt:ws}
    \underset{\mathbf{w}, s}{\min}:& \quad 
    \frac{1}{2}
    \left( \mathbf{A}_1\mathbf{w}\mathbf{c}^\top \mathbf{w} - 
    \mathbf{b} \right)^\top 
    \mathbf{\Lambda}
    \left( \mathbf{A}_1\mathbf{w}\mathbf{c}^\top \mathbf{w} - 
    \mathbf{b} \right) \\
    \text{s.t.} &\quad 
    (s^2 + \mathbf{w}^\top \mathbf{A}_2\mathbf{w})^2 = 1, \\
    & \quad  \mathbf{w}^\top \mathbf{w} \le s^2 \eta^2. \label{ineq:ws}
\end{align}
According to the definition of $\mathbf{A}_2$, it is positive semidefinite. Hence, we have $s^2 =1- \mathbf{w}^\top \mathbf{A}_2\mathbf{w} $. Plug in the expression of $s^2$ into \eqref{ineq:ws}, the constraints in problem~\eqref{opt:ws} can be simplified to $\mathbf{w}^\top (\mathbf{I}+\eta^2\mathbf{A}_2)\mathbf{w} \le \eta^2$. Assume $\mathbf{I}+\eta^2\mathbf{A}_2 = \mathbf{U}^\top\mathbf{U}$ is the Cholesky decomposition of $\mathbf{I}+\eta^2\mathbf{A}_2$. Define $\mathbf{H} = \mathbf{A}_1\mathbf{U}^{-1}$, $\mathbf{e} = \mathbf{U}^{-\top}\mathbf{c}$, and $\mathbf{x} = \mathbf{U}\mathbf{w}$, we can simplify problem~\eqref{opt:ws} further as:
\begin{align} \label{opt:s2final}
    \min_{\mathbf{x}}:& \quad 
    \frac{1}{2}\left( \mathbf{H}\mathbf{x}\mathbf{e}^\top\mathbf{x} - \mathbf{b}\right)^\top 
    \mathbf{\Lambda}
    \left( \mathbf{H}\mathbf{x}\mathbf{e}^\top\mathbf{x} - \mathbf{b}\right) \\\nonumber
    \text{s.t.}&\quad 
    \mathbf{x}^\top\mathbf{x} \le \eta^2. 
\end{align}

\begin{algorithm}[t!]
	\caption{Optimal Adversarial Data Point Design while Making Small Changes to Other Regression Coefficients }\label{alg:poly-attack}
	\begin{algorithmic}[1]
		\State \textbf{Input}: the data set, $\{y_i, \mathbf{x}_i\}_{i=1}^n$, energy budget $\eta$, and the index of feature to be attacked, the trade-off parameter $\lambda$. 
		\State \textbf{Steps}:
		\State compute $\mathbf{A}$ according to equation~\eqref{eq:A}, compute $\boldsymbol{\beta}_0$ according to~\eqref{eq:beta0}, compute $\mathbf{A}_2$ according to~\eqref{eqs:a1a1}. 
		\State follow the steps~\eqref{prob:gener-orig},~\eqref{prob:quad},~\eqref{opt:s23}, and~\eqref{opt:ws}, and formulate our problem as a polynomial optimization problem~\eqref{opt:s2final}.
		
		\State  use Lasserre's relaxation method to solve problem~\eqref{opt:s2final} and get the optimal solution $\mathbf{x}^*$ and optimal value $p^*$. 
		\State compute $\mathbf{w}^* = \mathbf{U}^\top\mathbf{x}^*$, where $\mathbf{I} +\eta^2\mathbf{A}_2 = \mathbf{U}\mathbf{U}^\top$. 
		\State compute $s^*=\pm\sqrt{ 1-(\mathbf{w}^*)^\top\mathbf{A}_2\mathbf{w}^*}$. 
		\State calculate the optimal solution $\mathbf{x}_0^* = \mathbf{w}_{1:m}^*/s^*, y_0^* =w_{m+1}^*/s^*$. 
		\State \textbf{Output}: return the optimal adversarial data point $\{y_0^*, \mathbf{x}_0^* \}$ and the optimal value $p^*$. 
	\end{algorithmic}
\end{algorithm}

This is an optimization problem with a quartic objective function and with quadratic constraint. 
Recent progress in multivariate polynomial optimization has made it possible to solve this problem using the sum of squares technology \cite{lasserre2001global,laurent2009sums,weisser2018sparse,wainwright2006log}. This method finds the global optimal solutions by solving a sequence of convex linear matrix inequality problems. Even though this sequence might be infinitely long, in practice, a very short sequence is enough to guarantee its global optimality. Hence, in this subsection, we will resort to Lasserre's relaxation method \cite{lasserre2001global}. In Appendix~\ref{app:poly}, we briefly discuss how to use this method to solve~\eqref{opt:s2final}. Algorithm~\ref{alg:poly-attack} summarizes the process to design the adversarial data point.
Numerical examples using this method to solve our problem with real data will be provided in Section \ref{sec:ne}.  

\section{Rank-one attack analysis}\label{sec:atk-r1}
In Section~\ref{sec:one-data-attack}, we have discussed how to design one adversarial data point to attack the regression coefficients. In this section, we consider a more powerful adversary who can modify the whole dataset in order to attack the regression coefficients. In particular, we will consider a rank-one attack on the feature matrix \cite{li2020trans}. This type of attack covers many useful scenarios, for example, modifying one entry of the feature matrix, deleting one feature, changing one feature or replacing one feature etc. Hence, the analysis of the rank-one attack provides a universal bound for all of these kinds of modifications. 
Specifically, we will consider the objective in problem~\eqref{opt:orig-min} and~\eqref{opt:orig-max} where the adversary attacks one particular regression coefficient. In the following, we will first formulate our problem and then provide our alternative optimization method to solve this problem. 


 In the considered rank one attack model, the attacker will carefully design a rank-one feature modification matrix $\mathbf{\Delta}$ and add it to the original feature matrix $\mathbf{X}$. As the result, the modified feature matrix is $\hat{\mathbf{X}} = \mathbf{X} + \mathbf{\Delta}$. As $\mathbf{\Delta}$ has rank one, we can write $\mathbf{\Delta}= \mathbf{c}\mathbf{d}^\top$. Similar to the previous section, we restrict the adversary to having constrained energy budget, $\eta$. Here, we use the Frobenius norm  to measure the energy of the modification matrix. Hence, we have $\|\mathbf{\Delta}\|_{\text{F}} \le \eta $. If the attacker's goal is to increase the importance of feature $i$, our problem can be written as 
\begin{align}\label{opt:r1-max}
    \max_{\|\mathbf{c}\mathbf{d}^\top\|_{\text{F}}\le \eta}:& \quad|\hat{\beta}_i| \\\nonumber
    \text{s.t.}& \quad \hat{\boldsymbol{\beta}} = \underset{\boldsymbol{\beta}}{\text{argmin}}\, \|\mathbf{y} - \hat{\mathbf{X}}\boldsymbol{\beta} \|^2, \\\nonumber
    &\quad \hat{\mathbf{X}} = \mathbf{X}+ \mathbf{c}\mathbf{d}^\top.
\end{align}
If the adversary is trying to minimize the magnitude of the $i$th regression coefficient, our problem is 
\begin{align}\label{opt:r1-min}
    \min_{\|\mathbf{c}\mathbf{d}^\top\|_{\text{F}}\le \eta}:& \quad|\beta_i| \\ \nonumber
    \text{s.t.}& \quad \hat{\boldsymbol{\beta}} = \underset{\boldsymbol{\beta}}{\text{argmin}}:\, \|\mathbf{y} - \hat{\mathbf{X}}\boldsymbol{\beta} \|^2, \\\nonumber
    &\quad \hat{\mathbf{X}} = \mathbf{X}+ \mathbf{c}\mathbf{d}^\top.
\end{align}
Similar to Section~\ref{sec:atk-one},  the solutions to problems~\eqref{opt:r1-max} and~\eqref{opt:r1-min} can be obtained by the solutions to the following two problems:
\begin{align}\label{prob-max}
    \max_{\|\mathbf{c}\mathbf{d}^\top\|_{\text{F}}\le \eta}:& \quad \hat{\beta}_i \\\nonumber
    \text{s.t.}& \quad \hat{\boldsymbol{\beta}} = \underset{\boldsymbol{\beta}}{\text{argmin}}:\, \|\mathbf{y} - \hat{\mathbf{X}}\boldsymbol{\beta} \|^2, \\\nonumber
    &\quad \hat{\mathbf{X}} = \mathbf{X}+ \mathbf{c}\mathbf{d}^\top.
\end{align}
and 
\begin{align}\label{prob-min}
    \min_{\|\mathbf{c}\mathbf{d}^\top\|_{\text{F}}\le \eta}:& \quad \hat{\beta}_i \\\nonumber
    \text{s.t.}& \quad \hat{\boldsymbol{\beta}} = \underset{\boldsymbol{\beta}}{\text{argmin}}:\, \|\mathbf{y} - \hat{\mathbf{X}}\boldsymbol{\beta} \|^2, \\\nonumber
    &\quad \hat{\mathbf{X}} = \mathbf{X}+ \mathbf{c}\mathbf{d}^\top.
\end{align}
We can further write the above two problems in a more unified form:
\begin{align}\label{prob-unif}
    \min_{\|\mathbf{c}\mathbf{d}^\top\|_{\text{F}}\le \eta}:& \quad \mathbf{e}^\top \hat{\boldsymbol{\beta}} \\\nonumber
    \text{s.t.}& \quad \hat{\boldsymbol{\beta}} = \underset{\boldsymbol{\beta}}{\text{argmin}}:\, \|\mathbf{y} - \hat{\mathbf{X}}\boldsymbol{\beta} \|^2, \\\nonumber
    &\quad \hat{\mathbf{X}} = \mathbf{X}+ \mathbf{c}\mathbf{d}^\top.
\end{align}
If $\mathbf{e} = \mathbf{e}_i$, in which $ \mathbf{e}_i$ is a vector with the $i$th entry being $1$ and all other entries being zero, problem~\eqref{prob-unif} is equivalent to problem~\eqref{prob-min}. If $\mathbf{e} = -\mathbf{e}_i$, problem~\eqref{prob-unif} is equivalent to problem~\eqref{prob-max}. Hence, in the following part, we will focus on solving this unified problem~\eqref{prob-unif}. 

To solve problem~\eqref{prob-unif}, we can first solve the lower level optimization problem in the constraints, where it admits a simple solution that $\hat{\boldsymbol{\beta}} = \hat{\mathbf{X}}^\dagger\mathbf{y}$ and $\mathbf{X}^\dagger$ is the pseudo-inverse of $\mathbf{X}$. Different from the adding one data point case discussed in Section~\ref{sec:one-data-attack}, this pseudo-inverse does not have a simple expression in terms of $\mathbf{c}$, $\mathbf{d}$ and $\mathbf{X}$.
However, this pseudo-inverse can be written as
$\hat{\mathbf{X}}^\dagger = \mathbf{X}^\dagger + \mathbf{G}$, 
where $\mathbf{G}$ depends on $\mathbf{c}$, $\mathbf{d}$, and $\mathbf{X}$. To write the expression of $\hat{\mathbf{X}}^\dagger$ in a more compact way, we first introduce the following notations:
\begin{align} \nonumber
    &\gamma = 1+\mathbf{d}^\top\mathbf{X}^\dagger\mathbf{c},
    & \mathbf{v} = \mathbf{X}^\dagger\mathbf{c}, \\ \nonumber
    &\mathbf{n} = (\mathbf{X}^\dagger)^\top\mathbf{d},
    &\mathbf{w} = (\mathbf{I}-\mathbf{X}\mathbf{X}^\dagger)\mathbf{c}.
\end{align}
Furthermore, we assume that the feature matrix has full column rank. Depending on the values of $\mathbf{w}$ and $\gamma$, the expression of $\mathbf{G}$ can be divided into the following four cases \cite{petersen2008matrix}:\\
\textbf{Case 1}: when $\|\mathbf{w}\| =0$, $\gamma \neq 0$, 
\begin{align}\label{r1-case1}
\mathbf{G} = 
- \frac{1}{\gamma}\mathbf{v}\mathbf{n}^\top;
\end{align}
\textbf{Case 2}: when $\|\mathbf{w}\|\neq 0$, $\gamma = 0$,
\begin{align} \label{r1-case2}
\mathbf{G} =
-\frac{1}{\|\mathbf{n}\|^2}\mathbf{X}^\dagger\mathbf{n}\mathbf{n}^\top 
-\frac{1}{\|\mathbf{w}\|^2}\mathbf{v}\mathbf{w}^\top;
\end{align}
\textbf{Case 3}: when $\|\mathbf{w}\| \neq 0$, $\gamma \neq 0$, 
\begin{align}\label{r1-case3}\nonumber
\mathbf{G} &= \frac{1}{\gamma}\mathbf{X}^\dagger\mathbf{n}\mathbf{w}^\top \\
&- \frac{\gamma}{\|\mathbf{n}\|^2\|\mathbf{w}\|^2+\gamma^2}
\left( \frac{\|\mathbf{w}\|^2}{\gamma}\mathbf{X}^\dagger\mathbf{n}+\mathbf{v}\right)
\left( \frac{\|\mathbf{n}\|^2}{\gamma}\mathbf{w}+\mathbf{n} \right)^\top;
\end{align}
\textbf{Case 4}: when $\|\mathbf{w}\|=0$, $\gamma =0 $,
\begin{align}\label{r1-case4}
\mathbf{G} = 
-\frac{1}{\|\mathbf{v}\|^2}\mathbf{v}\mathbf{v}^\top \mathbf{X}^\dagger
-\frac{1}{\|\mathbf{n}\|^2}\mathbf{X}^\dagger\mathbf{n}\mathbf{n}^\top + \frac{\mathbf{v}^\top \mathbf{X}^\dagger\mathbf{n}}{\|\mathbf{v}\|^2\|\mathbf{n}\|^2}\mathbf{v}\mathbf{n}^\top.
\end{align}
Since $\hat{\boldsymbol{\beta}} = \hat{\mathbf{X}}^\dagger \mathbf{y} = (\mathbf{X}^\dagger + \mathbf{G})\mathbf{y}$ and $\mathbf{X}^\dagger$ does not depend on $\mathbf{c}$ and $\mathbf{d}$, our problem is equivalent to 
\begin{align}\label{opt:r1-sim}
    \min_{\mathbf{c}, \mathbf{d}}:&\quad \mathbf{e}^\top \mathbf{G} \mathbf{y}\\ \nonumber
    \text{s.t.} & \quad \|\mathbf{c}\cdot\mathbf{d}^\top \|_{\text{F}} \le \eta.
\end{align}
Suppose $(\mathbf{c}^*, \mathbf{d}^*)$ is the optimal solution of~\eqref{opt:r1-sim}, it is easy to see that for nonzero $k$, $(k\mathbf{c}^*, \mathbf{d}^*/k)$ is also a valid optimal solution. To avoid the ambiguity, it is necessary and possible to further reduce the feasible region. Hence, we put an extra constraint on $\mathbf{c}$, where we restrict the norm of $\mathbf{c}$ to be less than or equal to $1$. As the result, our problem can be further written as
\begin{align}\label{opt:r1-sim2}
    \min_{\mathbf{c}, \mathbf{d}}:&\quad \mathbf{e}^\top \mathbf{G} \mathbf{y}\\\nonumber
    \text{s.t.} & \quad \|\mathbf{c}\| \le 1, \\ \nonumber
    &\quad \|\mathbf{d} \| \le \eta,
\end{align}
in which we use the identity $\|\mathbf{c}\mathbf{d}^\top\|_{\text{F}}=\|\mathbf{c}\|\|\mathbf{d}\|$. It is clear that problem~\eqref{opt:r1-sim} and problem~\eqref{opt:r1-sim2} have the same optimal objective value. 

Since $\mathbf{G}$ is decided by $\mathbf{c}$, $\mathbf{d}$, and $\mathbf{X}$, different values of $\mathbf{c}$ and $\mathbf{d}$ may result in different objective functions. As we have seen the expression of $\mathbf{G}$ can be divided into four different cases,  we will discuss these cases one by one in the following. 

Before further discussion, let us assume the singular value decomposition of the original feature matrix is $\mathbf{X} =\mathbf{U}\mathbf{\Sigma}\mathbf{V}^\top $, where $\mathbf{\Sigma} = [\text{diag}(\sigma_1,\sigma_2,\cdots,\sigma_m), \mathbf{0}]^\top$ and $\sigma_1 \ge \sigma_2 \ge \cdots \ge \sigma_m>0$. With this decomposition, we have $\mathbf{X}^\dagger =\mathbf{V} \mathbf{\Sigma}^\dagger\mathbf{U}^\top $, where $\mathbf{\Sigma}^\dagger = [\text{diag}(\sigma_1^{-1},\sigma_2^{-1},\cdots,\sigma_m^{-1}),\mathbf{0}]$. 

In Case 1~\eqref{r1-case1}, if $\eta\geq\sigma_m$, by letting $\gamma \to 0$, we have our objective being minus infinity by setting $(\mathbf{c}, \mathbf{d})=(\mathbf{u}_m, -\sigma_m \mathbf{v}_m)$ or $(\mathbf{c}, \mathbf{d})=(-\mathbf{u}_m, \sigma_dm \mathbf{v}_m)$, where $\mathbf{u}_m$ and $\mathbf{v}_m$ is the $m$th column of matrices $\mathbf{U}$ and $\mathbf{V}$, respectively. Hence, we conclude that, when $\eta\ge \sigma_m$, the optimal value of problem~\eqref{opt:r1-sim2} is unbounded from below. As the result, throughout this section, we assume $\eta <\sigma_m$. Thus, we have $\gamma = 1+ \mathbf{d}^\top \mathbf{X}^\dagger \mathbf{c} \ge 1 - \| \mathbf{c}\cdot\mathbf{d}^\top\|\|\mathbf{X}^\dagger\| \ge 1 - \frac{\eta}{\sigma_m} >0$. As $\gamma \neq 0$, we don't need to consider Case 2 and Case 4. Actually, Case 1 is a special case of Case 3. In particular, setting $\mathbf{w}=\mathbf{0}$ in Case 3, we recover Case 1. Based on these discussion, we only need to solve the problem in Case 3.
Let $h$ denote our objective $h(\mathbf{c},\mathbf{d}) = \mathbf{e}^\top \mathbf{G} \mathbf{y}$, plug in the expression of $\mathbf{G}$, and we have
\begin{align}
    \nonumber 
    h(\mathbf{c},\mathbf{d}) =& \frac{1}{\|\mathbf{n}\|^2\|\mathbf{w}\|^2+\gamma^2}
    \big(
    \gamma\mathbf{e}^\top \mathbf{X}^\dagger\mathbf{n}\mathbf{w}^\top \mathbf{y}
    -\gamma\mathbf{e}^\top \mathbf{v}\mathbf{n}^\top \mathbf{y} \\
    &-\|\mathbf{w}\|^2\mathbf{e}^\top\mathbf{X}^\dagger\mathbf{n}\mathbf{n}^\top \mathbf{y}
    -\|\mathbf{n}\|^2\mathbf{e}^\top \mathbf{v}\mathbf{w}^\top \mathbf{y}
    \big).
    \label{obj:h}
\end{align}
We need to optimize $h(\mathbf{c},\mathbf{d})$ over $\mathbf{c}$ and $\mathbf{d}$ with the constraint $\|\mathbf{c}\|\le 1$ and $\|\mathbf{d}\|\le \eta$. However, $h(\mathbf{c},\mathbf{d})$ is a ratio of two quartic functions, which is known to be hard non-convex problem in general. To solve this problem, similar to \cite{mei2015using}, we can use the projected gradient descent method (Please see Appendix~\ref{appx:pgd} for details). However, it is hard to choose a proper stepsize and its convergence is not clear when the projected gradient descent is used on a non-convex problem. In the following, we provide an alternating optimization algorithm with provable convergence guarantee.

The enabling observation of our approach is that, even though the optimization problem is a complex non-convex problem, for a fixed $\mathbf{c}$, $h$ is a ratio of two quadratic functions with respect to $\mathbf{d}$. Similarly, for a fixed $\mathbf{d}$, $h$ is a ratio of two quadratic functions with respect to $\mathbf{c}$. Ratio of two quadratic functions admits a hidden convex structure \cite{beck2009convex}. Inspired by this, we decompose our optimization variables into $\mathbf{c}$ and $\mathbf{d}$, and then use alternating optimization algorithm described in Algorithm~\ref{alg:ao} to sequentially optimize $\mathbf{c}$ and $\mathbf{d}$.

\begin{algorithm}[t]
\caption{Optimal Rank-one Attack Matrix Design via the Alternating Optimization Algorithm}\label{alg:ao}
\begin{algorithmic}[1]
\State \textbf{Input}: data set $\{y_i, \mathbf{x}_i\}_{i=1}^n$ and energy budget $\eta$.
\State \textbf{Initialize}: randomly initialize $\mathbf{c}^0$ and $\mathbf{d}^0$, set number of iterations $k=0$.
\State compute $\mathbf{G}$ according to~\eqref{r1-case3}.
\State plug in the expression of $\mathbf{G}$ into~\eqref{opt:r1-sim2}, and obtain our objective, $h(\mathbf{c}, \mathbf{d})$, as in~\eqref{obj:h}. 
\State \textbf{Do}
\State update $\mathbf{c}^k$ by solving: $\mathbf{c}^k = \underset{\|\mathbf{c}\|\le 1}{\text{argmin}}:  h(\mathbf{c},\mathbf{d}^{k-1}),$

\State update $\mathbf{d}^k$ by solving: $\mathbf{d}^k = \underset{\|\mathbf{d}\|\le \eta }{\text{argmin}}:  h(\mathbf{c}^k, \mathbf{d}),$ 

\State set $k = k+1$,
\State \textbf{While} convergence conditions are not meet.
\State compute the modification matrix $\mathbf{\Delta} =\mathbf{c}^k(\mathbf{d}^k)^\top$.
\State \textbf{Output}: return the modification matrix, $\mathbf{\Delta}$.
\end{algorithmic}
\end{algorithm}

The core of this algorithm is to solve the following two problems 
\begin{align}
    \label{opt-c}
    \mathbf{c}^k = \underset{\|\mathbf{c}\|\le 1}{\text{argmin}}: h(\mathbf{c},\mathbf{d}^{k-1}), 
\end{align}
and 
\begin{align}
    \label{opt-d}
    \mathbf{d}^k = \underset{\|\mathbf{d}\|\le \eta}{\text{argmin}}: h(\mathbf{c}^k,\mathbf{d}).
\end{align}
For a fixed $\mathbf{d}$, the objective of problem~\eqref{opt-c} becomes 
$h(\mathbf{c},\mathbf{d}) = {h_1(\mathbf{c})}/{h_2(\mathbf{c})},$
where we omit the superscript of $\mathbf{d}$, 
\begin{align}\nonumber
    &h_1(\mathbf{c})=
    \mathbf{c}^\top 
    \big[
    \mathbf{e}^\top \mathbf{X}^\dagger \mathbf{n}\mathbf{n}\mathbf{y}^\top(\mathbf{I} - \mathbf{X}\mathbf{X}^\dagger)
    -\mathbf{n}^\top\mathbf{y}\mathbf{n}\mathbf{e}^\top \mathbf{X}^\dagger \\ \nonumber
    &-\mathbf{e}^\top\mathbf{X}^\dagger\mathbf{n}\mathbf{n}^\top \mathbf{y}(\mathbf{I}-\mathbf{X}\mathbf{X}^\dagger) 
    -\|\mathbf{n}\|^2(\mathbf{X}^\dagger)^\top\mathbf{e}\mathbf{y}^\top (\mathbf{I} - \mathbf{X}\mathbf{X}^\dagger)
    \big]\mathbf{c} \\ 
    &+\big[
    \mathbf{e}^\top \mathbf{X}^\dagger\mathbf{n}(\mathbf{I-\mathbf{X}\mathbf{X}^\dagger})\mathbf{y} 
    -\mathbf{n}^\top \mathbf{y}(\mathbf{X}^\dagger)^\top\mathbf{e}
    \big]^\top\mathbf{c} , \label{f1c}
\end{align}
and 
\begin{align}
    h_2(\mathbf{c}) = &
    \mathbf{c}^\top \big[
    \|\mathbf{n}\|^2(\mathbf{I}-\mathbf{X}\mathbf{X}^\dagger)
    +\mathbf{n}\mathbf{n}^\top 
    \big]\mathbf{c}
    +2\mathbf{n}^\top \mathbf{c} +1.\label{f2c}
\end{align}
Hence, problem~\eqref{opt-c} can be written as:
\begin{align} \label{prob-qrc}
    \min_{\mathbf{c}}:& \quad \frac{h_1(\mathbf{c})}{h_2(\mathbf{c})} \\
    \text{s.t.}&\quad \|\mathbf{c}\| \le 1, 
\end{align}
where the forms of $h_i(\mathbf{c}) = \mathbf{c}^\top \mathbf{A}_i\mathbf{c} + 2\mathbf{b}_i^\top \mathbf{c}+ l_i,\, i=1,2$ and $\mathbf{A}_i$, $\mathbf{b}_i$ and $l_i$ can be derived from~\eqref{f1c} and~\eqref{f2c}. The objective of this problem is the ration of two quadratic functions. Even though it is non-convex, it has certain hidden convex structures. The following theorem characterizes its optimal solution by solving a semidefinite programming \cite{beck2009convex}. 

\begin{theorem}\label{thm1}
The optimal value of problem~\eqref{prob-qrc} is the same as the following optimal value
\begin{align}\label{thm1:opt-dual}
    \max_{\alpha, \,\nu \ge 0}:&\quad \alpha \\ \nonumber
    \text{s.t.}&\quad 
    \begin{bmatrix}
    \mathbf{A}_1  & \mathbf{b}_1 \\ 
    \mathbf{b}_1^\top & l_1 
    \end{bmatrix}
    \succeq
    \alpha
    \begin{bmatrix}
    \mathbf{A}_2   & \mathbf{b}_2 \\
    \mathbf{b}_2^\top & l_2
    \end{bmatrix}
    -\nu 
    \begin{bmatrix}
    \mathbf{I} & \mathbf{0} \\
    \mathbf{0}   &  -1 
    \end{bmatrix}
\end{align}
\end{theorem}
\begin{proof}
First, we will show that there exists $\mu \ge 0$ such that 
\begin{align} \label{thm1-cond}
\begin{bmatrix}
\mathbf{A}_2      & \mathbf{b}_2 \\
\mathbf{b}_2^\top & l_2 
\end{bmatrix}
+ \mu
\begin{bmatrix}
\mathbf{I} & \mathbf{0} \\
\mathbf{0}  & -1
\end{bmatrix}
\succ \mathbf{0}. 
\end{align}	
	
To prove the left hand side of \eqref{thm1-cond} is positive definite, we can show the following two inequalities are true according to Schur complement condition for positive definite matrix 
\begin{align} \label{thm1-p1}
    l_2-\mu >0, 
    \\ 
    \mathbf{A}_2 + \mu\mathbf{I}- \frac{1}{1-\mu}\mathbf{b}_2\mathbf{b}_2^\top >0\label{thm1-p2},
\end{align}
where $l_2=1$. Plug in the expression of $\mathbf{A}_2$, the left hand of inequality~\eqref{thm1-p2} can be written as
\begin{align*}
    &\mathbf{A}_2 + \mu\mathbf{I}-\frac{1}{1-\mu}\mathbf{b}_2\mathbf{b}_2^\top  \\
    &= 
    \|\mathbf{n}\|^2(\mathbf{I} - \mathbf{X}\mathbf{X}^\dagger)
    +\mu\mathbf{I} - \frac{\mu}{1-\mu}\mathbf{n}\mathbf{n}^\top.
\end{align*}
Since $\mathbf{I}-\mathbf{X}\mathbf{X}^\dagger$ is a projection matrix, it is positive semi-definite. So, we only need to prove
\begin{align}
    \mu\mathbf{I} -\frac{\mu}{1-\mu}\mathbf{n}\mathbf{n}^\top \succ \mathbf{0}.
\end{align}
Since $\mathbf{n}\mathbf{n}^\top$ is rank-one and its non-zero eigenvalue is $\|\mathbf{n}\|^2$, it equals to proving
$
    \|\mathbf{n}\|^2/(1-\mu) < 1.
$
To guarantee this inequality, we only need to make sure $\mu < 1-\|\mathbf{n}\|^2.$
Since $\|\mathbf{X}^\dagger\| \le 1/\sigma_m$ and $\|\mathbf{d}\| \le \eta$, we get
$
    \|\mathbf{n}\|^2
    = \|(\mathbf{X}^\dagger)^\top\mathbf{d}\|^2
    \le \|\mathbf{X}^\dagger\|^2 \|\mathbf{d}\|^2 
    \le {\eta^2}/{\sigma_m^2}
    < 1.
$
By choosing $0 <  \mu <1-\|\mathbf{n}\|^2 < 1$, we can ensure \eqref{thm1-p1} and \eqref{thm1-p2} are both satisfied, and hence inequality~\eqref{thm1-cond} is satisfied. 

As inequality~\eqref{thm1-cond} holds, we have $h_2(\mathbf{c}) > 0$. It can be verified by the fact that left multiplying $[\mathbf{c}^\top, 1]$ and right multiplying $[\mathbf{c}^\top, 1]^\top$ result in $h_2(\mathbf{c}) > \mu(1-\mathbf{c}^\top\mathbf{c})>0$. So, our objective, $h(\mathbf{c}, \mathbf{d})$, is well defined. Using the same technique that transforms problem~\eqref{opt:qcrq-orig} to~\eqref{opt:qcrq}, we can further convert~\eqref{prob-qrc} to a quadratic constrained quadratic programming. Further analysis reveals when~\eqref{thm1-cond} holds, the feasible set of this quadratic constraint quadratic programming is compact. So, the minimum is attainable. Thus, we can solve it through is its dual problem~\eqref{thm1:opt-dual}. 
For the rest of the proof, please refer to \cite{beck2009convex} for details.
\end{proof}
From Theorem~\ref{thm1}, we know the optimal value of~\eqref{prob-qrc}. We now discuss how to find optimal $\mathbf{c}$ to achieve this value. 
Suppose the optimal solution of problem~\eqref{thm1:opt-dual} is $(\alpha^*, \nu^*)$. Since, $h_2(\mathbf{c}) >0$, we have $h_1(\mathbf{c}) \ge \alpha^* {h_2(\mathbf{c})} $ for any feasible $\mathbf{c}$. 
Hence, we can compute the optimal solution of problem~\eqref{prob-qrc} by solving
\begin{align}\label{thm1-prime}
    \underset{\mathbf{c}}{\text{argmin}}:&\quad h_1(\mathbf{c}) - \alpha^* h_2(\mathbf{c}) \\
    \text{s.t.}& \quad \|\mathbf{c}\|^2 \le 1
\end{align}
This problem is just a trust region problem. There are several existing methods to solve it efficiently. In this paper, we employ the method describe in \cite{beck2006finding}. 

Now, we turn to solve problem~\eqref{opt-d}. Since~\eqref{opt-d} and~\eqref{opt-c} have similar structure, we can employ the methods described in Theorem~\ref{thm1} and \eqref{thm1-prime} to find its optimal value and optimal solution for problem~\eqref{opt-d}.

Until now, we have fully described how to solve the intermediate problems in the alternating optimization method. The following theorem shows that the proposed alternating optimization algorithm will converge. Suppose the generated sequence of solution is $\{\mathbf{c}^k,\,\mathbf{d}^k\},\, k=0,1,\cdots$, and we have the following theorem:
\begin{theorem}
The sequence $\{\mathbf{c}^k,\,\mathbf{d}^k\}$ admits a limit point $\{\bar{\mathbf{c}}\,,\bar{\mathbf{d}}\}$ and we have 
\begin{align}\label{thm:limit}
\lim_{k\to\infty} h(\mathbf{c}^k,\mathbf{d}^k) = h(\bar{\mathbf{c}},\,\bar{\mathbf{d}}).
\end{align}
Furthermore, every limit point is a critical point, which means 
\begin{align}\label{thm:critic}
\nabla h(\bar{\mathbf{c}},\bar{\mathbf{d}})^\top 
\begin{bmatrix}
\mathbf{c} - \bar{\mathbf{c}} \\
\mathbf{d} - \bar{\mathbf{d}}
\end{bmatrix}
\ge 0, 
\end{align}
for any $\| \mathbf{c}\| \le 1$ and $\| \mathbf{d} \| \le \eta$. 
\end{theorem}

\begin{proof}
We first give the proof of $\eqref{thm:limit}$. 
Since the sequence $\{\mathbf{c}^k, \mathbf{d}^k\}$ lies in the compact sets, $\{ (\mathbf{c}, \mathbf{d})\,|\, \|\mathbf{c}\| \le 1,\, \|\mathbf{d}\| \le \eta \}$, and according to the Bolzano-Weierstrass theorem \cite{bartle2000introduction}, $\{\mathbf{c}^k, \mathbf{d}^k\}$ must have limit points. So, there is a subsequence of $\{h^k\}$ which converges to $h(\bar{\mathbf{c}},\bar{\mathbf{d}})$. As the objective is a continuous function with respect to $\mathbf{c}$ and $\mathbf{d}$, the compactness of the constraint also implies the sequence of the objective value, $\{h^k\}$, is bounded from below. In addition, $\{h^k\}$ is a non-increasing sequence, which indicates that the sequence of the function value must converge. In summary, the sequence $\{h^k\}$ must converge to $h(\bar{\mathbf{c}},\bar{\mathbf{d}})$. 
For the rest of the proof, please refer to Corollary 2 of \cite{grippo2000convergence} for more details.
\end{proof} 

\section{Numerical Examples}\label{sec:ne}

In this subsection, we test our adversarial attack strategies on a practical regression problem. In this regression task, we use seven international indexes to predict the returns of the Istanbul Stock Exchange~\cite{akbilgic2014novel}. The data set contains 536 data samples, which are the records of the returns of Istanbul Stock Exchange with seven other international indexes starting from Jun. 5, 2009 to Feb. 22, 2011. 

\subsection{Attacking one specific regression coefficient}

\begin{figure}[t!]
\begin{minipage}[b]{0.48\linewidth}
  \centering
  \centerline{\includegraphics[width=\linewidth]{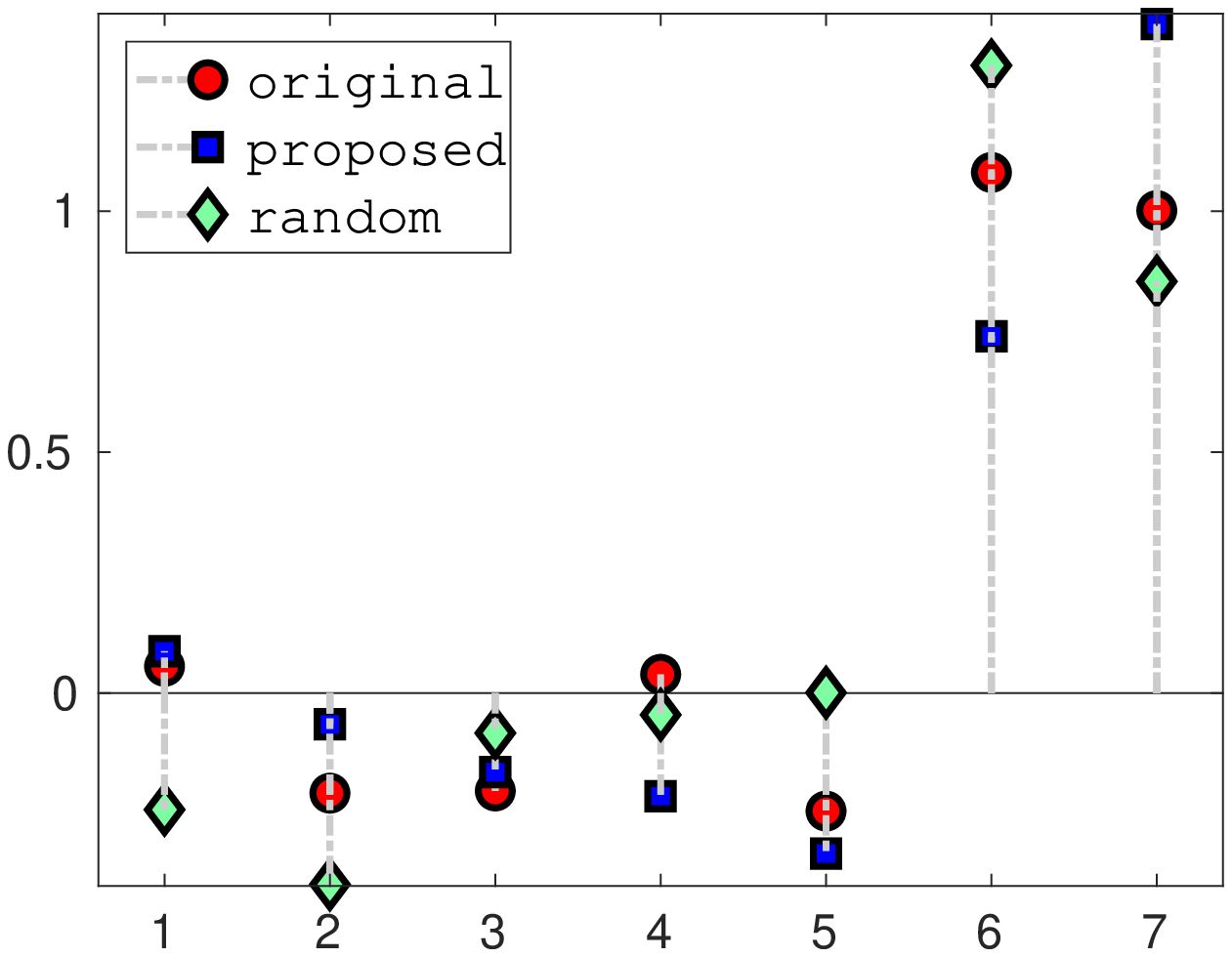}
  }

  \centerline{(a)}\medskip
\end{minipage}
\hfill
\begin{minipage}[b]{0.48\linewidth}
  \centering
  \centerline{\includegraphics[width=\linewidth]{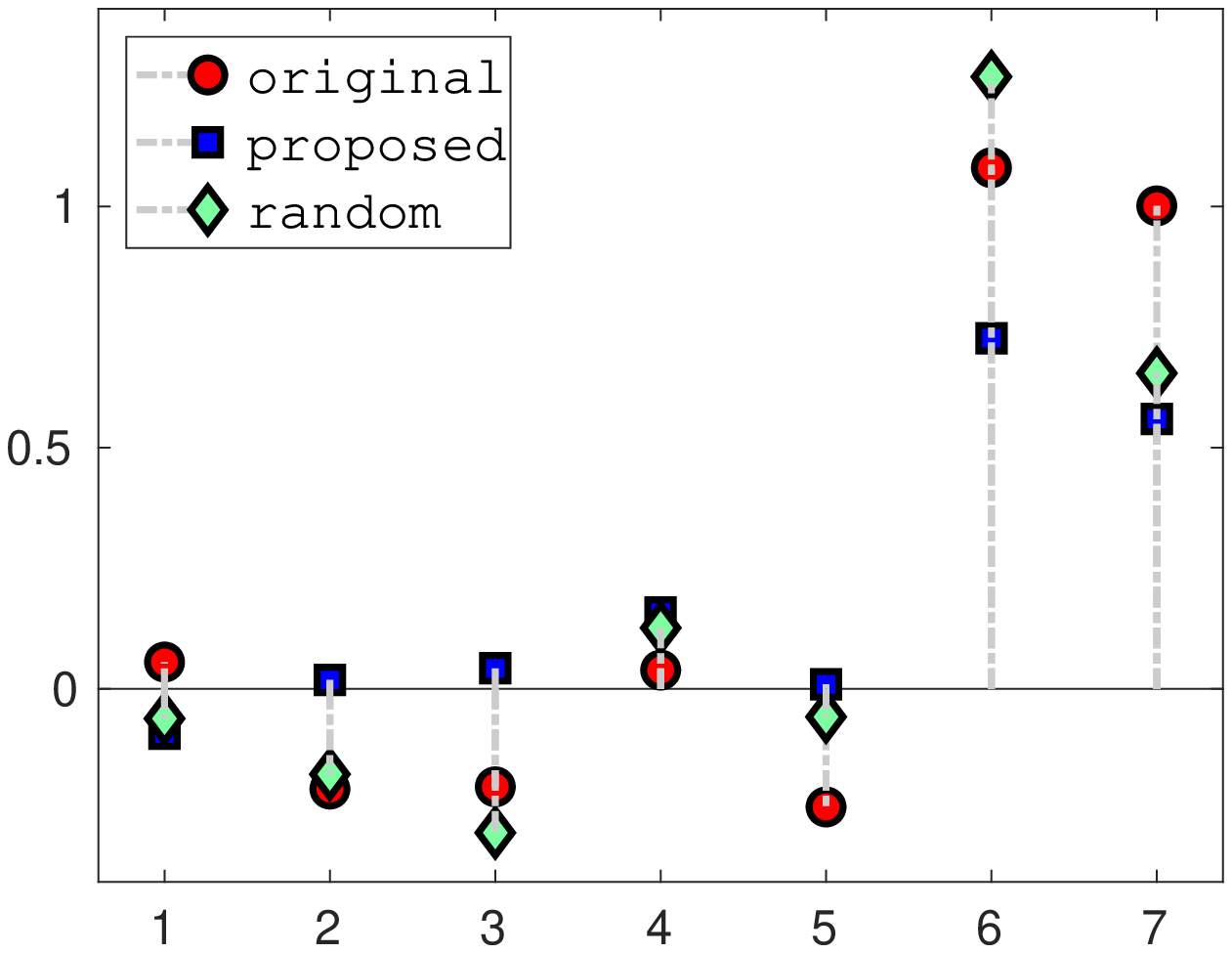}}
  \centerline{(b)}\medskip
\end{minipage}
\caption{The regression coefficients under our proposed and random attacks. The left figure shows the regression coefficients before and after attacking the fourth coefficient with objective~\eqref{opt:min-beta} and the right one shows the regression coefficients before and after attacking the fifth coefficient with objective~\eqref{opt:orig-min}.}
 \label{fig:exp}
\end{figure}

We first show the results for attacking one specific regression coefficient. The results are shown in Fig.~\ref{fig:exp}. In the figure, the $x$-axis denotes the index of the regression coefficients and the $y$-axis indicates the value of the regression coefficients. 
We design our first experiment to attack the fourth regression coefficient and try to make it small by solving problem~\eqref{opt:min-beta}.  We use two strategies to attack this coefficient with fixed energy budget $\eta=0.2$. The first strategy is the one proposed in this paper. As a comparison, we also use a random strategy. In the random strategy, we randomly generate the adversarial data point with each entry being i.i.d. generated from a standard normal distribution. Then, we normalize its energy to be $\eta$. We repeat this random attack $10000$ times and select the one with the smallest value of the fourth regression coefficient. 
In the second experiment, we intend to make the absolute value of the fifth regression coefficient small. We compare the proposed and the random attack strategies to attack the fifth coefficient with fixed energy budget $\eta=0.1$. Similarly, for the random attack strategy, we run $10000$ times random attacks and select the one with the smallest absolute value of the fifth regression coefficient. 

Fig.~\ref{fig:exp} (a) shows the regression coefficients before and after the first experiment and Fig.~\ref{fig:exp} (b) shows the regression coefficients before and after the second experiment. From the figures we can see that our proposed adversarial attack strategy is much more efficient than the random attack strategy. One can also observe that by only adding one adversarial example, designed using the approach characterized in this paper, one can dramatically change the value of a regression coefficient and hence change the importance of that explanatory variable. 

\subsection{Attacking without changing untargeted regression coefficients too much}\label{sec:addone}
\begin{figure}[t]
    \centering
    \includegraphics[width=.7\linewidth]{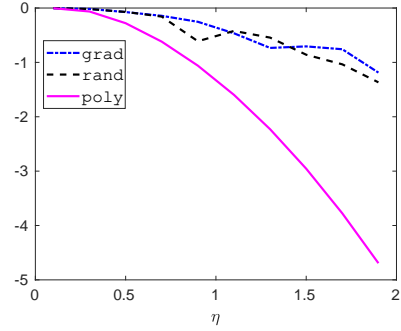}
    \caption{Attack the fourth regression coefficient with objective~\eqref{prob:gener-orig} and with different energy budgets. }
    \label{fig:muti-ind4etas}
\end{figure}

\begin{figure}[t]
    \centering
    \includegraphics[width=0.7\linewidth]{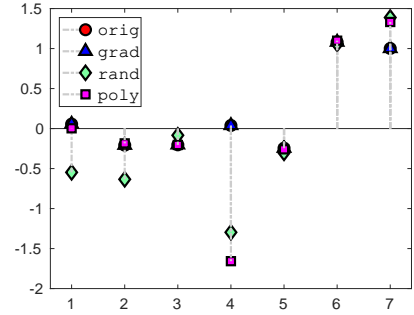}
    \caption{The regression coefficients after different kinds of strategies that attack the fourth regression coefficient with energy budget $\eta = 1$.}
    \label{fig:muti-ind4}
\end{figure}

From the numerical examples in the previous subsection, we can see the untargeted regression coefficients may change dramatically while we attacking one specific regression coefficient with an adversarial data point. For example,  as demonstrated in Fig.~\ref{fig:exp}, the sixth and seventh regression coefficients change significantly when we attack the fourth and fifth regression coefficients. To mitigate the undesirable changes of untargeted regression coefficients, we need more sophisticated attacking strategies. In this subsection, we will test different strategies with a more general objective function as demonstrated in Section~\ref{sec:atk-multi}. We also use the same data set as described in the previous subsection. We first try to attack the fourth regression coefficient to increase its importance while making only small changes to the rest of the regression coefficients. To accomplish this task, we aim to solve problem~\eqref{prob:gener-orig} with $\lambda = -1$. Given the energy budget, firstly, we use our semidefinite relaxation based algorithm to solve problem~\eqref{opt:s2final}, and then follow Algorithm~\ref{alg:poly-attack} to find the adversarial data point. For comparison, we also carry out the random attack strategy, in which we randomly generate the data point with each entry being i.i.d. according to the standard normal distribution. Then, we normalize its energy being $\eta$ and added it to the original data points. We repeat these random attacks $10000$ times and select the one with the smallest objective value. The third strategy is the projected gradient descent based strategy, where we use projected gradient descent algorithm to solve~\eqref{opt:s2final} and follow similar steps of Algorithm~\ref{alg:poly-attack} to find the adversarial data point. 
Projected gradient descent works much like the gradient descent except with an additional operation that projects result of each step onto the feasible set after moving in the direction of negative gradient \cite{parikh2014proximal}. 
We have described the general projected gradient descent algorithm in Appendix~\ref{appx:pgd} and Algorithm~\ref{alg:pgd}.
In our experiment, we use diminishing stepsize, $\alpha_t = 1/(t+1)$.
Since the projected gradient descent algorithm depends on the initial points heavily, given the energy budget, we repeat it $100$ times with different random initial points and treat the average of its objective values as the objective value of this algorithm.  

Fig.~\ref{fig:muti-ind4etas} shows the objective values under different energy budgets with different attacking strategies and Fig.~\ref{fig:muti-ind4} demonstrates the regression coefficients after one of the attacks of different strategies with $\eta=1$. In these figures, `orig' is the original regression coefficient, `grad' denotes the projected gradient descent algorithm, `rand' means the random strategy, and `poly' indicates the our semidefinite relaxation strategy. From these two figures, we can see our semidefinite relaxation based strategy performs much better than the other two strategies. In addition, in our experiment, our semidefinite relaxation method with relaxation order $2$ or $3$ can always lead to the global optimal solutions. Hence, the computational complexity of this method is still low. Fig.~\ref{fig:muti-ind4} also shows our relaxation based method leads to the largest magnitude of the fourth regression coefficient while keeping other regression coefficients almost unchanged.

\begin{figure}[t]
    \centering
    \includegraphics[width=.7\linewidth]{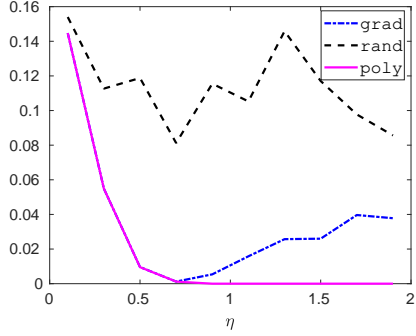}
    \caption{Attack the sixth regression coefficient with objective~\eqref{prob:gener-orig} under different energy budgets.}
    \label{fig:muti-ind6etas}
\end{figure}

\begin{figure}[t!]
    \centering
    \includegraphics[width=0.63\linewidth]{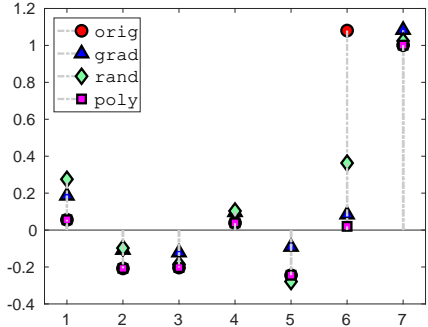}
    \caption{The regression coefficients after different kinds of strategies that attack the sixth regression coefficient with energy budget $\eta = 1$.}
    \label{fig:muti-ind6}
\end{figure}

In the second experiment, we attack the sixth regression coefficient and attempt to make its magnitude small while keeping the change of the rest of coefficients to be small. So, we set $\lambda = 1$ in problem~\eqref{prob:gener-orig} to achieve this goal. The settings of each strategy is similar to the ones in the first experiment. Fig.~\ref{fig:muti-ind6etas} shows the objective values with different strategies under different energy budgets and Fig.~\ref{fig:muti-ind6} demonstrates the regression coefficients after one of the attacks of those strategies respectively with energy budget $\eta =1$. From Fig.~\ref{fig:muti-ind6etas} we know the projected gradient descent based strategy and semidefinite relaxation based strategy achieve much lower objective values compared to the random attack strategy. Specifically, when the energy budget is smaller than $0.7$, both of the two strategies behave similarly. However, when the energy budget is larger than $0.7$, the projected gradient descent based strategy leads to larger objective value as the energy budget grows. This is because the projected gradient descent algorithm tends to find solutions at the boundary of the feasible set.  
By contrast, our semidefinite relaxation based strategy can find the global optimal solutions with relaxation order $2$ or $3$. Thus, it gives the best performance among the three strategies. Fig.~\ref{fig:muti-ind6} also demonstrates our relaxation based method achieves the global optimum when $\eta=1$ as the it leads the sixth regression coefficient to zero and other regression coefficients to be unchanged.  

\subsection{Rank-one attack}
\begin{figure*}
\begin{minipage}[b]{.3\linewidth}
\centering 
\centerline{
\includegraphics[width=\linewidth]{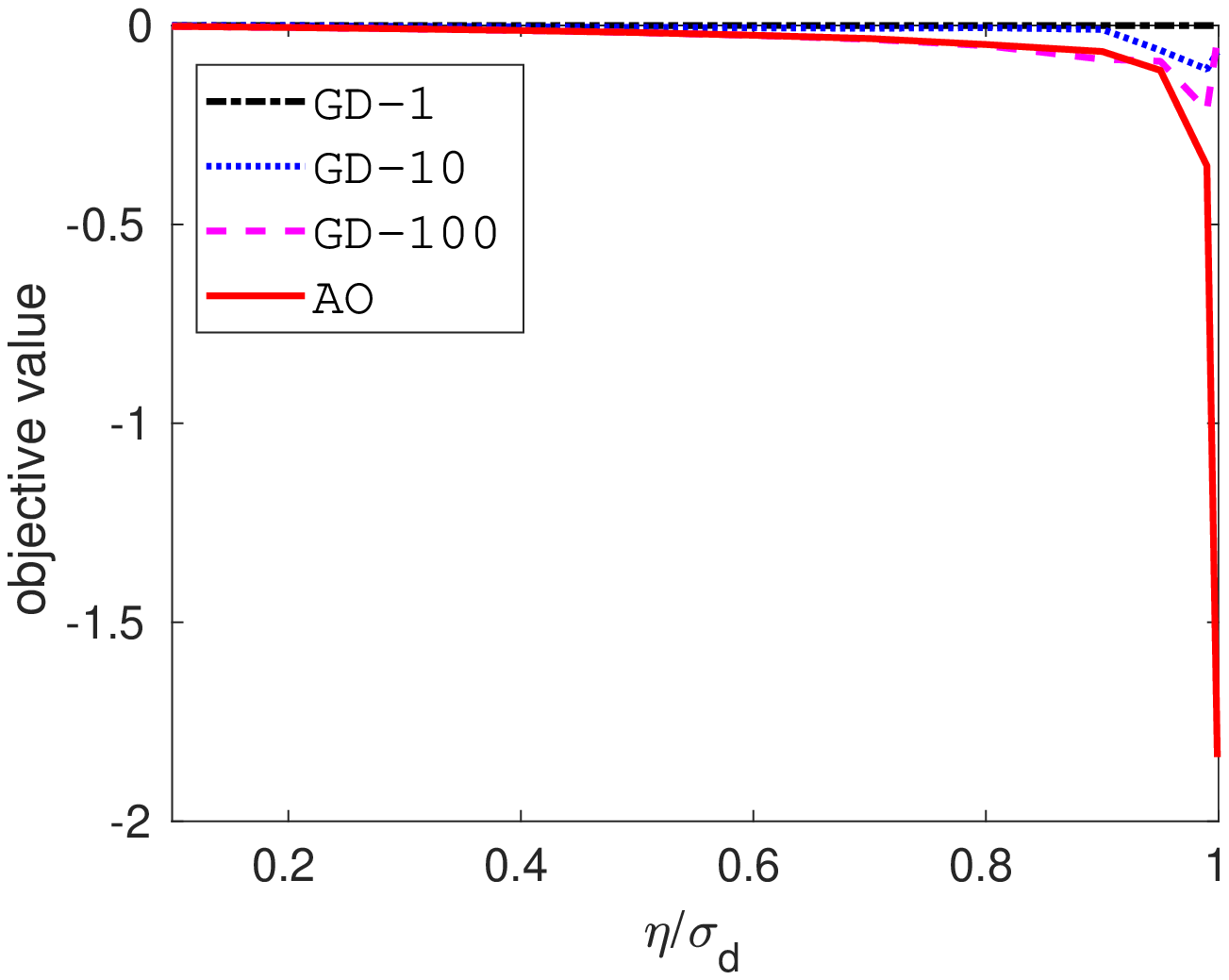}
}
\centerline{(a)}\medskip 
\end{minipage}
\hfill 
\begin{minipage}[b]{.3\linewidth}
\centering 
\centerline{
\includegraphics[width=\linewidth]{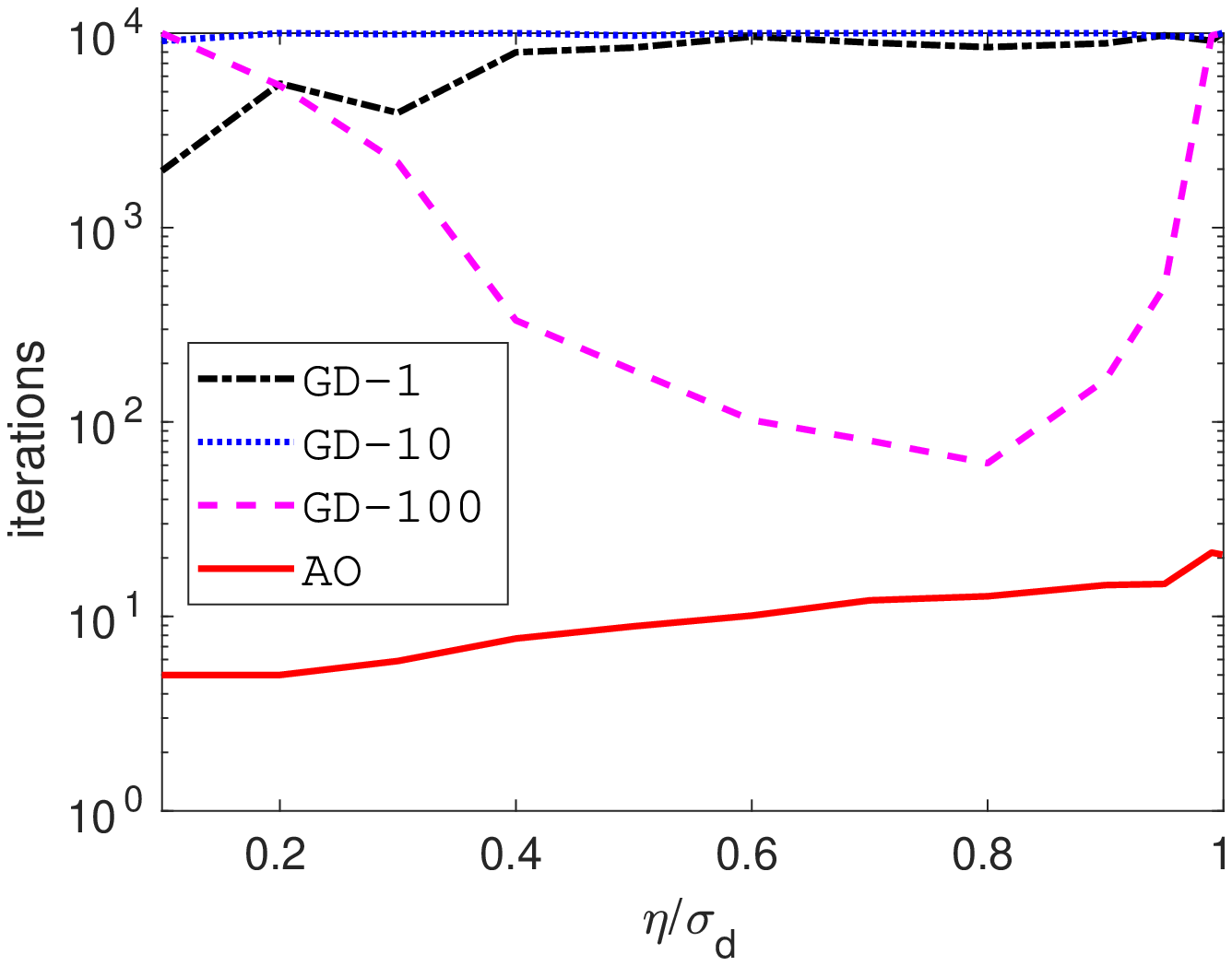}
}
\centerline{(b)}\medskip 
\end{minipage}   
\hfill 
\begin{minipage}[b]{.3\linewidth}
\centering 
\centerline{
\includegraphics[width=\linewidth]{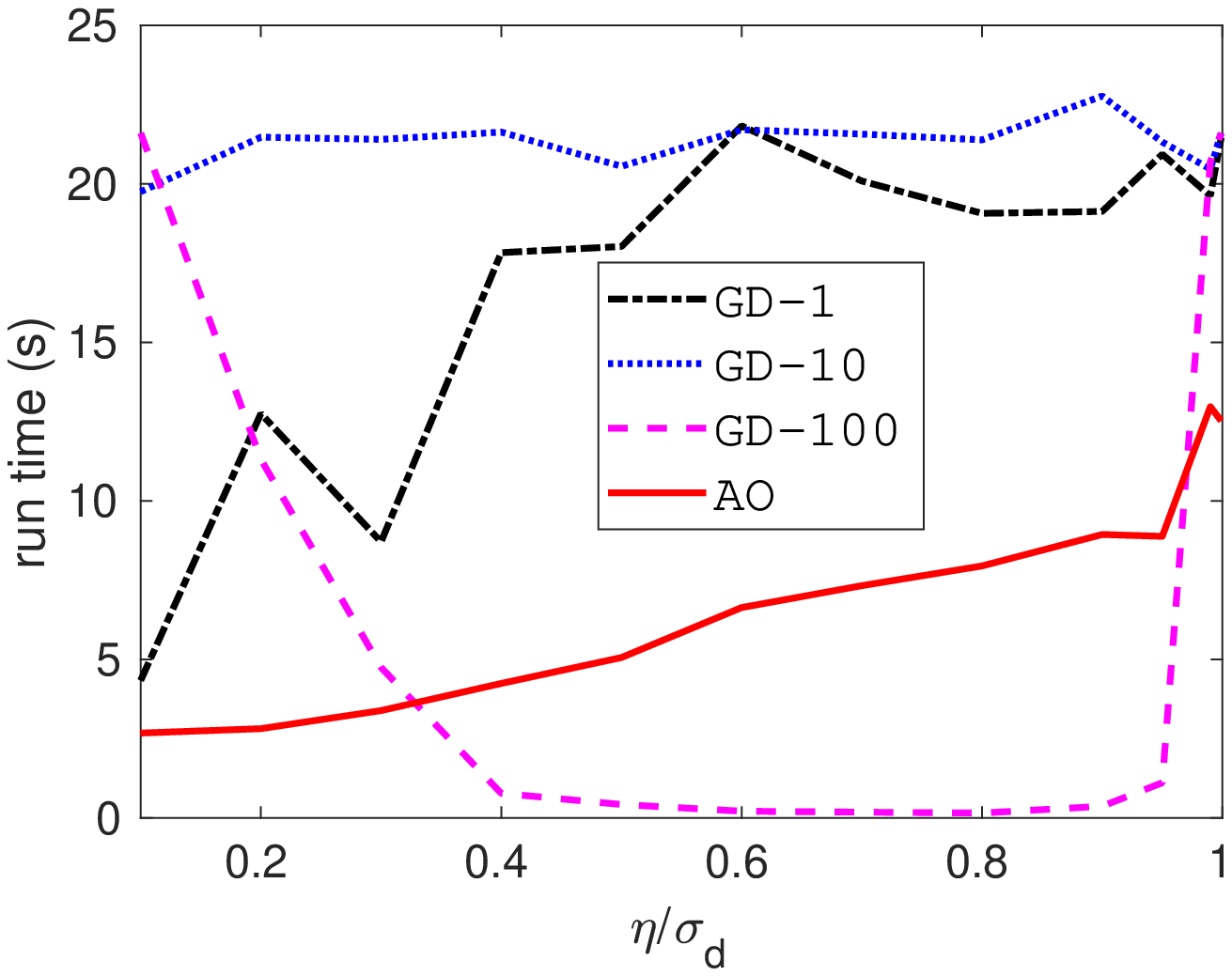}
}
\centerline{(c)}\medskip 
\end{minipage} 
\caption{Comparison of the projected gradient descent with different stepsizes and the proposed alternating optimization method. (a) shows the average objective values of different algorithms under different energy budget. (b) demonstrates the average iterations that reach the convergence condition and (c) illustrates the average run times of each algorithm.}
\label{fig:r1-values}
\end{figure*}

\begin{figure*}[t]
\begin{minipage}[b]{.3\linewidth}
\centering 
\centerline{
\includegraphics[width=\linewidth]{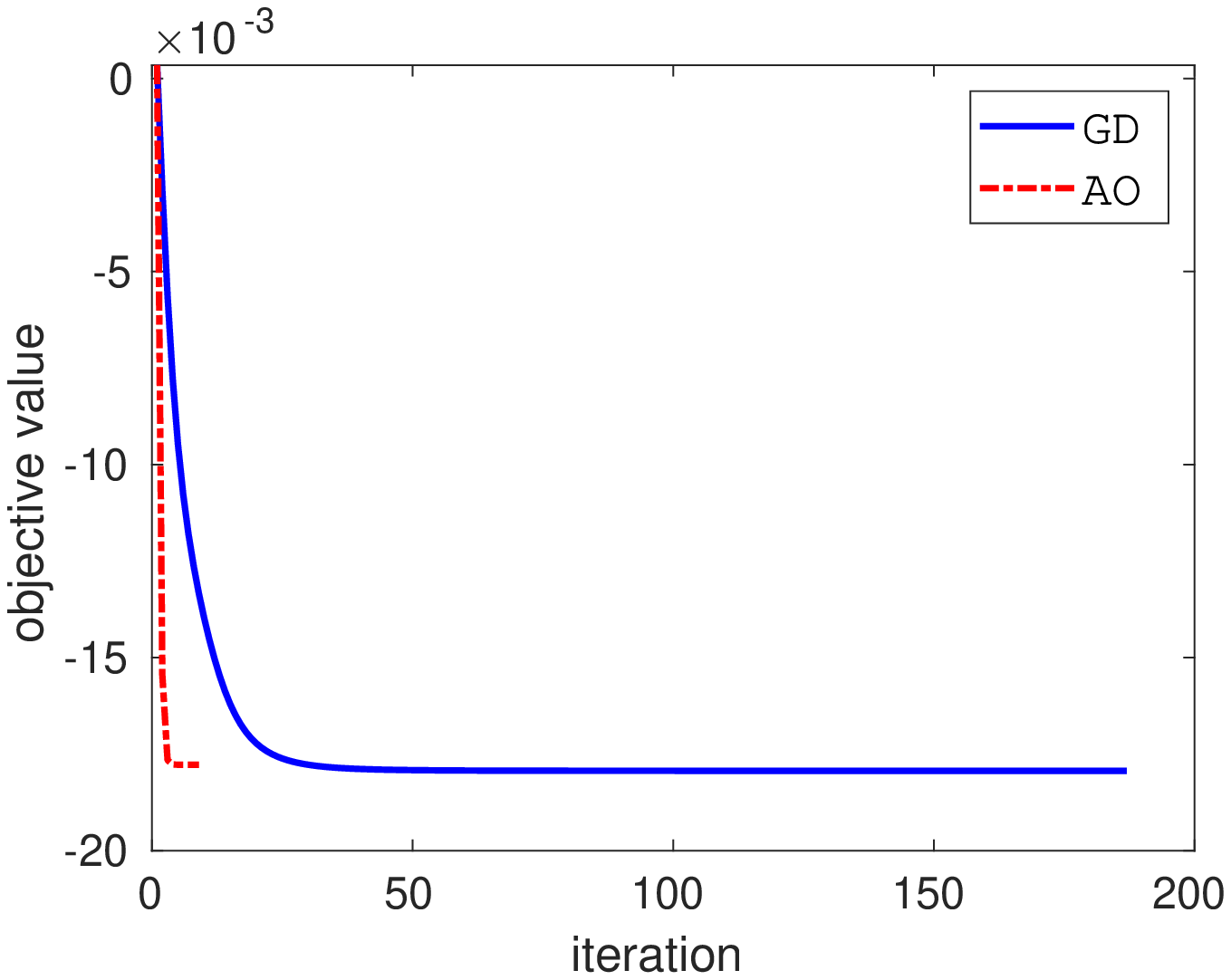}
}
\centerline{(a)}\medskip 
\end{minipage}
\hfill 
\begin{minipage}[b]{.3\linewidth}
\centering 
\centerline{
\includegraphics[width=\linewidth]{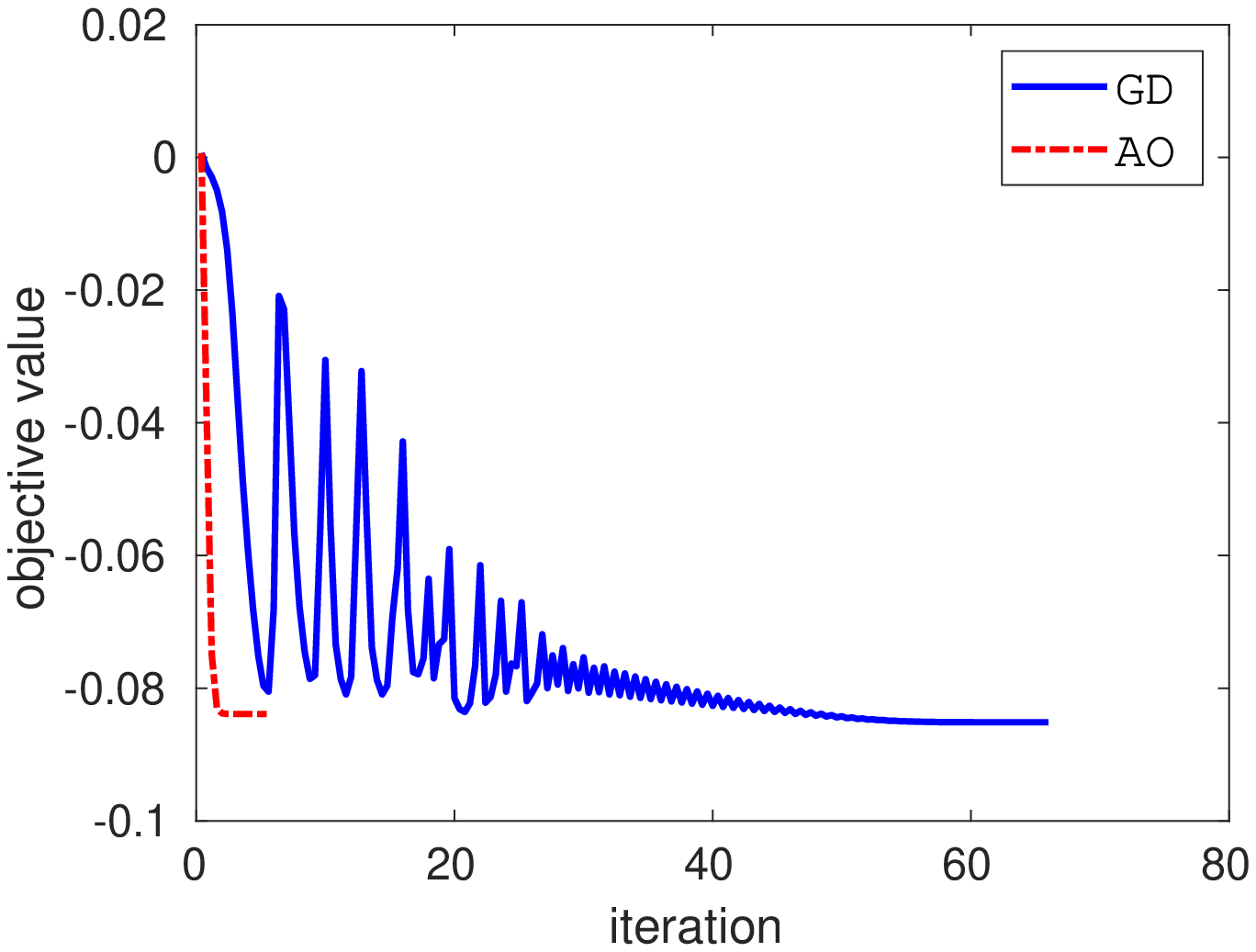}
}
\centerline{(b)}\medskip 
\end{minipage}   
\hfill 
\begin{minipage}[b]{.3\linewidth}
\centering 
\centerline{
\includegraphics[width=\linewidth]{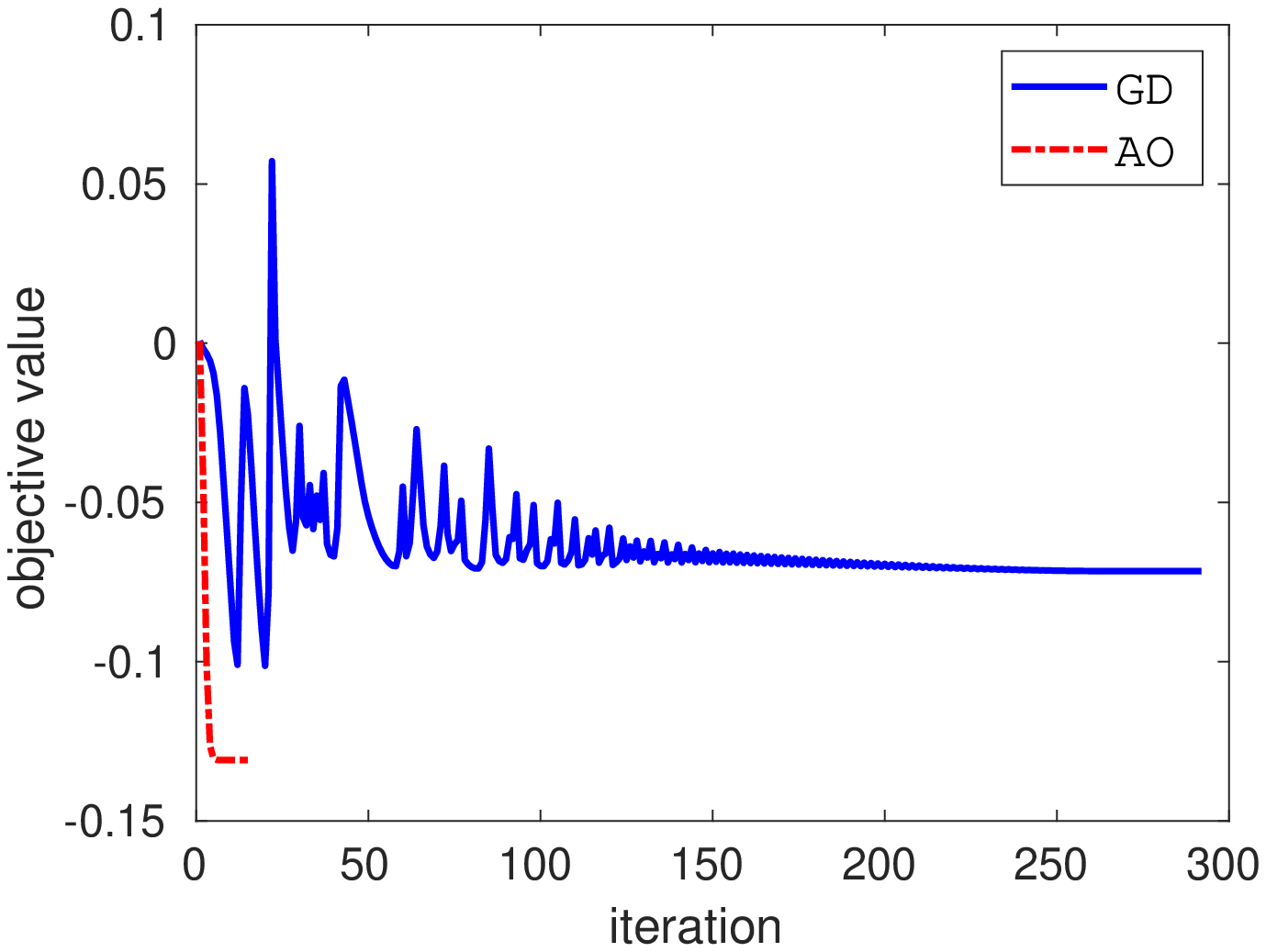}
}
\centerline{(c)}\medskip 
\end{minipage} 
\caption{These figures show the evolution of function values as the iteration increases with one typical run of projected gradient descent and alternating optimization algorithm, where (a) is with $\eta/\sigma_m = 0.5$, (b) is with $\eta/\sigma_m = 0.9$, and (c) is with $\eta/\sigma_m = 0.95$ and $\sigma_m$ is the smallest singular value of the original feature matrix.}
\label{fig:r1-conv}
\end{figure*}

In this subsection, we carry out different rank-one attack strategies. Our goal is to minimize the magnitude of the fourth regression coefficient with objective~\eqref{prob-min}. We compare two strategies: the projected gradient descent based strategy discussed in Section~\ref{sec:addone} and our proposed alternating optimization based strategy. 
For the projected gradient descent based strategies, we use different step sizes: $1/(1+t)$, $10/(1+t)$, and $100/(1+t)$. As our analysis shows, when the energy budget is larger than the smallest singular value, our objective can be minus infinity. Hence, in our experiment, we vary the energy budget from $0$ to the smallest singular value. Given certain energy budget, we set all the algorithms with the same randomly initialized point and run these algorithms until they stop with the same convergence condition that is two consecutive function values change too small or it reaches the maximal allowable iterations. We repeat this process $100$ times and record their average objective values, iterations, and run times when they converge. 

Fig.~\ref{fig:r1-values} shows the average objective values, the number of iterations, and the run time of the four algorithms, where GD-1, GD-10 and GD-100 stand for the projected gradient descent with stepsize $1/(1+t)$, $10/(1+t)$, and $100/(1+t)$, respectively, and AO denotes the proposed alternating optimization method. As Fig.~\ref{fig:r1-values} shows, when the energy budget increases, the objectives decrease for both of these algorithms. Furthermore, the proposed alternating optimization based algorithm provides much smaller objective values, especially when the energy budget approaches the smallest singular value. As the energy budget increases, the total iterations and run times of alternating optimization, GD-1, and GD-10 increase. One notable thing is that the proposed alternating optimization based algorithm converges several orders of magnitude faster than the projected gradient descent based algorithms. However, as the energy budget increases, the iterations and run times of GD-100 first decrease and then increase. This is due to the fact that a larger stepsize will result in a faster convergence rate while it may cause oscillation. This phenomena can be observed in Fig.~\ref{fig:r1-conv}, where it depicts the evolution of the objective values of AO and GD-100 with the energy budget being $\eta/\sigma_m=0.5$, $\eta/\sigma_m=0.9$ and $\eta/\sigma_m=0.95$, respectively, and $\sigma_m$ is the smallest singular value of the original feature matrix. From this figure we can see as the energy budget increases, the alternating optimization based algorithm converges very fast while the projected gradient descent based algorithm becomes unstable when the energy budget is large. 

\section{Conclusion}\label{sec:conclusion}
In this paper, we have investigated the adversarial robustness of linear regression problems. Particularly, we have given the closed-form solution when we attack one specific regression coefficient with limited energy budget. Furthermore, we have considered a more complex objective where we attack one of the regression coefficient while trying to keep the rest of regression coefficients to be unchanged. We have formulated this problem as a multivariate polynomial optimization problem and introduced the semidefinite relaxation method to solve it. Finally, we have studied a more powerful adversary who can make rank-one modification on the feature matrix. To take the advantage of the rank-one structure, we have proposed an alternating optimization algorithm to solve this problem. The numerical examples demonstrated that our proposed closed-form solution and the semidefinite relaxation based strategies can find the global optimal solutions and the alternating optimization based strategy provides better solutions, faster convergence, and more stable behavior compared to the projected gradient descent based strategy. 

In terms of future work, it is of interest to study how to design multiple adversarial data points and how to efficiently design the modification matrix without the rand-one constraint. 
\appendices
\section{Lasserre’s relaxation method for solving~\eqref{opt:s2final}}\label{app:poly}
In this appendix, we briefly introduce Lasserre's relaxation method and use this method to solve problem~\eqref{opt:s2final}. Lasserre's relaxation method is dedicated to solve multivariate polynomial optimization problem.   
A general multivariate polynomial optimization problem contains a multivariate polynomial objective function, $p(\mathbf{x}): \mathbb{R}^n\rightarrow \mathbb{R}$, and some constraints defined by polynomial inequalities, $g_i(\mathbf{x}) \ge 0$, $i=1,2,\ldots,r$:
\begin{align}\label{opt:poly-origi}
    \min:& \quad p(\mathbf{x}) \\
    \text{s.t.}& \quad g_i(\mathbf{x}) \ge 0,\, i=1,2,\ldots,r. 
\end{align}
Clearly, our optimization problem~\eqref{opt:s2final} can be viewed a multivariate polynomial optimization problem, since in \eqref{opt:s2final} the objective function is a fourth order multivariate polynomial and the constraint is a quadratic polynomial. 

To proceed, let us explain more details about the problem. The polynomial in the objective, $p(\mathbf{x})$, can be written as:
\begin{align}
    p(\mathbf{x}) = \sum_{\alpha}p_\alpha \mathbf{x}^{\alpha},
\end{align}
where $\alpha \in \mathbb{N}^n$, 
\begin{align}
    \mathbf{x}^\alpha = \prod_{i=1}^n x_i^{\alpha^i},
\end{align}
and $|\alpha| = \sum_i \alpha^i$. Suppose the order of the objective function is $m_0$, we have $|\alpha|\le m_0$. Define $\mathbf{p}_\alpha =\{ p_\alpha \} \in \mathbb{R}^{s(m_0)}$ as the coefficients of the polynomial basis
$\{1, x_1, x_2, \ldots, x_n, x_1^2, x_1x_2,\ldots, x_n^{m_0}\}$.
Hence, the dimension of the basis is $s(m_0) = \binom{n+m_0}{m_0}$. Instead of directly solving problem~\eqref{opt:poly-origi}, Lasserre's relaxation method \cite{lasserre2001global} first converts it into the following equivalent problem
\begin{align}\label{opt:poly-prob}
    \min_{\mu \in \mathcal{P}(\mathcal{K})}: \quad \int p(\mathbf{x})\, \mathrm{d}(\mu(\mathbf{x})), 
\end{align}
where $\mathcal{K}$ is the semialgebraic set defined by the inequalities:
$\mathcal{K}=\{ \mathbf{x}\,|\,g_i(\mathbf{x})\ge 0, \, i=1,2,\ldots, r \},$
and $\mathcal{P}(\mathcal{K})$ is the set of all probability measures supported on $\mathcal{K}$. 

To see that problem~\eqref{opt:poly-origi} and~\eqref{opt:poly-prob} are equivalent, suppose the optimal values of ~\eqref{opt:poly-origi} and ~\eqref{opt:poly-prob} are $p_0^*$ and $p^*$, respectively. Since $p(\mathbf{x})  \ge p^*_0$, we have $p^* \ge p^*_0$. Conversely, suppose the optimal solution of ~\eqref{opt:poly-origi} is $x^*$, $\mu = \delta_{x^*}$ is a feasible solution to~\eqref{opt:poly-prob}. Hence, we also have $p^* \le p^*_0$. Thus, the two problems are equivalent. 

With the help of this reformulation, finding the global optimal points for~\eqref{opt:poly-origi} is equivalent to finding the optimal distribution of~\eqref{opt:poly-prob}. Since $\int p(\mathbf{x})\, \mathrm{d}\mu(\mathbf{x}) = \sum_{\alpha}p_\alpha \int \mathbf{x}^\alpha\, \mathrm{d}\mu(\mathbf{x})$, the objective function of ~\eqref{opt:poly-prob} is just $\mathbf{p}_\alpha^\top \mathbf{y}_\alpha$, where $\mathbf{y}_\alpha =\{y_\alpha\}$ and  $y_\alpha = \int \mathbf{x}^\alpha \, \mathrm{d}\mu(\mathbf{x})$. So, finding the optimal probability is identical to finding the optimal $\mathbf{y}_\alpha$ under the constraint that $\mathbf{y}_\alpha$ is a valid moment sequence with respect to some probability measure on $\mathcal{K}$. The solution to this problem is fully characterized by the K-moment problem in case $\mathcal{K}$ is compact. Let us give more notations for the convenience of introducing this method. 

Given an $s(2m)$ length vector, $\mathbf{y}_\alpha = \{ y_\alpha \}$, with its first element $y_{0,\ldots,0}=1$. The $s(m)$ dimensional moment matrix $M_m(y)$ is constructed as follows: the first row and columns is defined as $M_m(1,k) = y_{\alpha_k}$ and ${M}_m(k,1) = y_{\alpha_k}$ for $k = 1,2,\ldots,s(m)$ and $M_m(i,j) = y_{\alpha_i+\alpha_j}$ for $i, j = 2,\ldots, s(m)$. For instance, when $n=2, m=2$, 
\begin{align*}
    M_m(y) =
    \begin{bmatrix}
    1      & y_{10} & y_{01} & y_{20} & y_{11} & y_{02} \\
    y_{10} & y_{20} & y_{11} & y_{30} & y_{21} & y_{12} \\
    y_{01} & y_{11} & y_{02} & y_{21} & y_{12} & y_{03} \\
    y_{20} & y_{30} & y_{21} & y_{40} & y_{31} & y_{22} \\
    y_{11} & y_{21} & y_{12} & y_{31} & y_{22} & y_{13} \\
    y_{02} & y_{12} & y_{03} & y_{22} & y_{13} & y_{04}
    \end{bmatrix}.
\end{align*}
Moreover, $M_m(y)$ defines a bi-linear form, $\langle \cdot\,, \cdot\rangle$, on two polynomials
\begin{align*}
    \langle p,q \rangle_{y} = \langle p, M_m(y)q \rangle 
    =\sum_\alpha (pq)_\alpha y_\alpha 
    =\int p(\mathbf{x})q(\mathbf{x})\,\mathrm{d}\mu(\mathbf{x}).
\end{align*}
So, if $\mathbf{y}_\alpha$ is a sequence of moments of some probability measure, we have 
\begin{align*}
    \langle q, q \rangle_y = \int q(x)^2 \,\mathrm{d}(\mu(\mathbf{x})) \ge 0.
\end{align*}
Thus, we have $M_m(y) \succcurlyeq 0$. 
Let $p(\mathbf{x}) $ be a multivariate polynomial with coefficient vector $\mathbf{p}_\beta =\{p_\beta\}$, and define the localizing matrix $M_m(py)$ as 
$$M_m(py)(i,j) = \sum_\beta p_\beta y_{\alpha_i+\alpha_j + \beta}.$$
For example, with 
\begin{align*}
    M_1(y) = 
    \begin{bmatrix}
    1      & y_{10} & y_{01} \\
    y_{10} & y_{20} & y_{11} \\
    y_{01} & y_{11} & y_{02}
    \end{bmatrix}
    \quad \text{and} \quad 
    p(\mathbf{x}) = a - x_1^2 - x_2^2, 
\end{align*}
we have 
\begin{align*}
    &M_1(py) =\\
    &\begin{bmatrix}
    a - y_{20} - y_{02}       & ay_{10} - y_{30} - y_{12} & ay_{01} - y_{21} - y_{03} \\
    ay_{10} - y_{30} - y_{12} & ay_{20} - y_{40} - y_{22} & ay_{11} - y_{31} - y_{13} \\
    ay_{01} - y_{21} - y_{03} & ay_{11} - y_{31} - y_{13} & ay_{01} - y_{22} - y_{04}
    \end{bmatrix}.
\end{align*}
Also, if $p(\mathbf{x})\ge 0$, by definition, we have $M_m(py) \succcurlyeq 0$. 


Further, we make the following assumption on the semialgebraic set $\mathcal{K}$. 
\begin{assumption} \label{assp:archimedean}
  The set $\mathcal{K}$ is compact and there exists a real-valued polynomial $u(\mathbf{x})$: $\mathbb{R}^n \rightarrow \mathbb{R}$ such that $\{u(\mathbf{x})\ge 0\}$ is compact and 
  \begin{align} \label{eq:assp}
  u(\mathbf{x}) = u_0(\mathbf{x})+ \sum_{k=1}^r g_i(\mathbf{x})u_i(\mathbf{x}) \ for \  all \ \mathbf{x} \in \mathbb{R}^n,
  \end{align}
  where the polynomial $u_i(\mathbf{x})$ is the sum of squares for $i=0,1,\ldots,r$.
\end{assumption}

Assumption~\ref{assp:archimedean} is satisfied in many cases. For example, this assumption is satisfied when there is only one inequality constraint that is compact, which is the case in our problem~\eqref{opt:s2final}. 

With the help of the notations and Assumption~\ref{assp:archimedean}, we have the main result. Let $w_i = \left\lceil m_i/2 \right\rceil$, where $m_i, i=1,2,\ldots, r,$ is the order of $g_i(\mathbf{x})$ and $m_0$ is the order of the objective, with $N \ge \max\{w_i\}$ for $i=0,1,\ldots,r$. Consider the following semidefinite programming
\begin{align}\label{opt:sdp}
    \min:&  \quad \sum_\alpha p_\alpha y_\alpha \\ \nonumber
    \text{s.t.}&  \quad M_N(y) \succcurlyeq 0, \\ \nonumber
    &\quad  M_{N-w_i}(g_iy) \succcurlyeq 0,\, i = 1,2,\ldots,r,
\end{align}
where $N$ is called the relaxation order. Lasserre \cite{lasserre2001global} shows that as $N$ approaches infinity, the solution of ~\eqref{opt:sdp} converges to the solution of \eqref{opt:poly-prob}. However, the dimension of the semidefinite programming~\eqref{opt:sdp} grows rapidly as $N$ increases and infinite $N$ makes solving problem~\eqref{opt:sdp} infeasible. Fortunately, in practice, a small $N$ is enough to get a very good approximation of problem~\eqref{opt:poly-prob} \cite{lasserre2001global}. Furthermore, a small $N$ is usually sufficient to get the global optimal solutions and the sufficient rank condition,
$\text{rank} M_{N}(y) = \text{rank} M_{N-w_{\max}}(y)$,
where $w_{\max} = \max \{w_i\}, i=0,1,\ldots,r$, assures the global optimality. Therefore, after we solving problem~\eqref{opt:sdp} we are ready to check whether we reach the global optimality. Besides, Henrion and Lasserre developed a systematic way to extract all the optimal solutions in case the rank condition is satisfied \cite{henrion2005detecting}. Since our problem~\eqref{opt:s2final} is just a special case of multivariate polynomial optimization, with the help of this relaxation method, we can solve problem~\eqref{opt:s2final}. 

\section{The Projected Gradient Descent Algorithm}\label{appx:pgd}
In this section, we describe the general projected gradient descent algorithm. Projected gradient descent algorithm is a way to solve the constrained optimization problem. So, it can be used to solve our problem as~\eqref{opt:s2final} and~\eqref{opt:r1-sim2} are both constrained optimization problems. 
At each step of projected gradient descent, we first compute the gradient of the objective function, move forward to the negative gradient direction, and then project onto the feasible set. We continue this process until meet the convergence criteria. When performing the projected gradient descent algorithm, we need to pay close attention to the stepsize parameter. Appropriate constant stepsize will lead to convergent algorithm if the objective is Lipschitz convex. However, diminishing stepsize is more proper for non-convex objective functions \cite{parikh2014proximal}. We summarize the general projected gradient descent method in Algorithm~\ref{alg:pgd}. 

\begin{algorithm}[t]
\caption{The Projected Gradient Descent Algorithm}\label{alg:pgd}
\begin{algorithmic}[1]
\State \textbf{Input}: objective function $f(\mathbf{x})$, feasible set $\mathcal{C}$, and stepsize parameter $\alpha_t$.
\State \textbf{Initialize}: randomly initialize $\mathbf{x}_0 \in \mathcal{C}$ and set the number of iterations $t=0$. 

\State \textbf{Do}
\State compute the gradient: $\nabla f(\mathbf{x}_t)$,

\State update:
$\mathbf{y}_t = \mathbf{x}_t - \alpha_t \nabla f(\mathbf{x}_t),$

\State project $\mathbf{y}_t$ onto the feasible set:
$\mathbf{x}_{t+1} = \underset{\mathbf{x}\in \mathcal{C}}{\text{argmin}}:\,\|\mathbf{x} - \mathbf{y}_t \|.$

\State set $t = t+1$,
\State \textbf{While} convergence conditions are not meet.
\State \textbf{Output}: $\mathbf{x}_t$.
\end{algorithmic}
\end{algorithm}

\bibliographystyle{./format/IEEEtran}
\bibliography{./format/IEEEabrv,mybib}
\end{document}